\documentclass[bj]{imsart}

\RequirePackage[OT1]{fontenc}
\RequirePackage{amsthm,amsmath}
\RequirePackage[colorlinks,citecolor=blue,urlcolor=blue]{hyperref}
\usepackage{graphicx}
\usepackage{amsmath}
\usepackage{mathrsfs}
\usepackage{amssymb}
\usepackage{float}
\usepackage{verbatim}
\usepackage{enumerate}

\newtheorem{theorem}{Theorem}
\newtheorem{lemma}{Lemma}

\newtheorem{corollary}{Corollary}

\newtheorem{definition}{Definition}

\newtheorem{remark}{Remark}
\arxiv{arXiv:1705.09898}

\startlocaldefs
\numberwithin{equation}{section}
\theoremstyle{plain}

\theoremstyle{remark}

\endlocaldefs

\begin{document}

\begin{frontmatter}
\title{Projection Theorems of Divergences and Likelihood Maximization Methods}
\runtitle{Projection Theorems of Divergences and Likelihood Maximization Methods}

\begin{aug}
\author{\fnms{ATIN} \snm{GAYEN}\ead[label=e1]{}}
\and
\author{\fnms{M. Ashok} \snm{Kumar}\ead[label=e2]{}}


\address{Discipline of Mathematics, Indian Institute of Technology Indore, Indore 453 553.
\printead{}}

\runauthor{ATIN GAYEN AND ASHOK KUMAR}

\affiliation{Discipline of Mathematics, Indian Institute of Technology Indore}
\address{Email: atinfordst@gmail.com \and ashokm@iiti.ac.in}

\end{aug}
\maketitle

\begin{abstract}{\textbf{Abstract:}\\\\}	
Projection theorems of divergences enable us to find reverse projection of a divergence on a specific statistical model as a forward projection of the divergence on a different but rather ``simpler'' statistical model, which, in turn, results in solving a system of linear equations. Reverse projection of divergences are closely related to various estimation methods such as the maximum likelihood estimation or its variants in robust statistics. We consider projection theorems of three parametric families of divergences that are widely used in robust statistics, namely the R\'enyi divergences (or the Cressie-Reed power divergences), density power divergences, and the relative $\alpha$-entropy (or the logarithmic density power divergences). We explore these projection theorems from the usual likelihood maximization approach and from the principle of sufficiency. In particular, we show the equivalence of solving the estimation problems by the projection theorems of the respective divergences and by directly solving the corresponding estimating equations. We also derive the projection theorem for the density power divergences.
\end{abstract}

\begin{keyword}
\kwd{Estimating equations, power divergence, power-law family, projection theorem, relative $\alpha$-entropy, R\'enyi divergence, robust estimation, sufficient statistics}
\end{keyword}

\end{frontmatter}

\section{Introduction}
\label{1sec:introduction}
Minimum divergence\footnote{By a divergence, we mean a non-negative extended real valued function defined for every pair of probability measures $(P,Q)$ satisfying $D(P,Q) = 0$ if and only if $P = Q$.} (or minimum distance) method has a unique place in statistical inference because
of its ability to trade-off between efficiency and robustness \cite{BasuSP11B,Pardo06B}. Minimization of Kullback-Leibler divergence ($I$-divergence) or relative entropy is closely related to the maximum likelihood estimation (MLE) \cite[Lemma~3.1]{CsiszarS04B}. MLE is not a preferred method when the data set is contaminated by some unexpected sample points called {\em outliers}. However, $I$-divergence can be generalized, replacing the logarithmic function by some power function, to produce divergences that are more robust with respect to outliers \cite{BasuHHJ98J, JonesHHB01J}. In this paper we consider three such families of divergences that are well-known in the context of robust statistics. They are defined as follows.

Let $\mathcal{X}$ be a finite set. Let $\mathcal{P} := \mathcal{P}(\mathcal{X})$ be the space of all strictly positive probability measures\footnote{$\mathcal{P}$ can be just thought of the set of all ordered $|\mathcal{X}|$-tuples with strictly positive components that sum to one.} on $\mathcal{X}$. Let $\alpha>0, \alpha\neq 1$, and let $P, Q\in \mathcal{P}$.
\begin{enumerate}
	\item The $D_\alpha$-divergence\footnote{Upto a monotone function, same as the Cressie-Read power divergence \cite{CressieR84J}.} (also known as R\'enyi divergence \cite{Renyi61J} or {\em power divergence} \cite{PatraMPB13J}):
	\begin{eqnarray}
	\label{1defn:D_alpha_divergence}
	D_\alpha(P,Q) := \frac{1}{\alpha-1}\log\sum\limits_{x\in\mathcal{X}}P(x)^\alpha Q(x)^{1-\alpha}.
	\end{eqnarray}
	\item The $B_\alpha$-divergence (also known as {\em density power divergence} \cite{BasuHHJ98J}):
	\begin{eqnarray}
	\label{1defn:B_alpha_divergence}
	B_{\alpha}(P,Q) := \frac{\alpha}{1 - \alpha}\sum\limits_{x\in \mathcal{X}}P(x)Q(x)^{\alpha -1} - \frac{1}{1 - \alpha}\sum\limits_{x\in \mathcal{X}}P(x)^{\alpha} + \sum\limits_{x\in \mathcal{X}}Q(x)^{\alpha}.
	\end{eqnarray}
	\item The $\mathscr{I}_\alpha$-divergence \cite{Sundaresan02ISIT}, \cite{Sundaresan07J}, \cite{LutwakYZ05J}, \cite{FujisawaE08J} (also known as {\em relative $\alpha$-entropy} \cite{KumarS15J1}, \cite{KumarS15J2}, {\em logarithmic density power divergence} \cite{MajiGB16J}):
	\begin{eqnarray}
	\label{1defn:I_alpha_divergence}
	\lefteqn{\mathscr{I}_\alpha(P,Q)}\nonumber\\
	&& \hspace*{-0.4cm} := \frac{\alpha}{1-\alpha}\log\sum\limits_{x\in\mathcal{X}}P(x)Q(x)^{\alpha-1}
	-\frac{1}{1-\alpha}\log\sum\limits_{x\in\mathcal{X}}P(x)^\alpha+\log\sum\limits_{x\in\mathcal{X}}
	Q(x)^\alpha.
	\end{eqnarray}
\end{enumerate}
Throughout the paper $\log$ stands for the natural logarithm. It should be noted that, although the above divergences are not defined for $\alpha=1$, they all coincide with the $I$-divergence as $\alpha\rightarrow 1$ \cite{CichockiA10J}. That is,
\begin{eqnarray}
\label{1eqn:kullback-leibler}
\displaystyle\lim_{\alpha\rightarrow 1} B_\alpha(P,Q) = \displaystyle\lim_{\alpha\rightarrow 1} \mathscr{I}_\alpha(P,Q) = \displaystyle\lim_{\alpha\rightarrow 1} D_\alpha(P,Q)= I(P,Q) := \sum\limits_{x\in\mathcal{X}}P(x)\log\frac{P(x)}{Q(x)}.
\end{eqnarray}
In this sense each of these three classes of divergences can be regarded as a generalization of the $I$-divergence.

$D_\alpha$-divergences also arise as generalized cut-off rates in information theory \cite{Csiszar95J1}. The $B_\alpha$-divergences belong to the Bregman class which is characterized by transitive projection rules (see \cite[Eq.~(3.2) and Th.~3]{Csiszar91J}). The $\mathscr{I}_\alpha$-divergence (for $\alpha < 1$) arises in information theory as redundancy in mismatched guessing moments \cite{Sundaresan02ISIT}, in mismatched compression \cite{KumarS15J2}, and in mismatched encoding of tasks \cite{BunteL14J}. The three classes of  divergences are associated with robust inference for $\alpha>1$ in case of $B_\alpha$ and $\mathscr{I}_\alpha$, and $\alpha<1$ in case of $D_\alpha$, as we shall see now.

Suppose that $X_1,\dots, X_n$ are an independent and identically distributed (i.i.d.) random sample drawn according to a particular member of a parametric family of probability measures, $\Pi = \{P_\theta : \theta\in\Theta\}\subset \mathcal{P}$, where $\Theta$ is an open subset of $\mathbb{R}^k$. To find the MLE of the parameter $\theta$, one needs to solve the so-called {\em score equation} or {\em estimating equation} for $\theta$, given by
\begin{eqnarray}
\label{1eqn:score_equation_mle_in_terms_of_sample}
\sum\limits_{j=1}^n s(X_j; \theta) = 0,
\end{eqnarray}
where $s(x;\theta) := \nabla\log P_\theta(x)$, called the {\em score function} and $\nabla$ stands for gradient with respect to $\theta$. Observe that (\ref{1eqn:score_equation_mle_in_terms_of_sample}) can be re-written as
\begin{eqnarray}
\label{1eqn:score_equation_mle}
\sum\limits_{x\in\mathcal{X}}\widehat{P}(x)s(x;\theta) = 0,
\end{eqnarray}
where $\widehat{P}(\cdot)$ is the empirical probability measure of the sample $X_1,\dots, X_n$.

In the presence of outliers in the observed sample, one modifies the score equation by scaling the score function in (\ref{1eqn:score_equation_mle}) by weights that down-weights the effect of outliers. The following estimating equation, referred as generalized Hellinger estimating equation, was proposed where the score function was weighted by $\widehat{P}(\cdot)^\alpha P_\theta(\cdot)^{1-\alpha}$ instead of $\widehat{P}(\cdot)$:
\begin{eqnarray}
\label{1eqn:score_equation_for_D_alpha}
\sum\limits_{x\in\mathcal{X}}\widehat{P}(x)^\alpha P_\theta (x)^{1-\alpha}s(x;\theta) = 0,
\end{eqnarray}
where $\alpha\in (0,1)$. The above estimating equation was proposed based on the following intuition. If $x$ is an outlier in the sample, $\widehat{P}(x)^\alpha$ is larger while $P_\theta(x)^{1-\alpha}$ is smaller for sufficiently smaller value of $\alpha$. Hence the terms in (\ref{1eqn:score_equation_for_D_alpha}) corresponding to outliers are down-weighted. (See \cite[Sec. 4.3]{BasuSP11B} and the references therein.)
The following estimating equation, where the score function is weighted by power of model probability measure and equated to its hypothetical one, was proposed by Basu {\em et. al.} \cite{BasuHHJ98J}:
\begin{eqnarray}
\label{1eqn:score_equation_B_alpha}
\sum\limits_{x\in\mathcal{X}}\widehat{P}(x) P_{\theta}(x)^{\alpha-1} s(x,\theta) = \sum\limits_{x\in\mathcal{X}}P_{\theta}(x)^{\alpha} s(x,\theta),
\end{eqnarray} 
where $\alpha>1$. Motivated by the works of Field and Smith \cite{FieldS94J} and Windham \cite{Windham95J}, further
an alternative estimating equation was proposed by Jones {\em et. al.} \cite{JonesHHB01J} where
the weights in (\ref{1eqn:score_equation_B_alpha}) were normalized:
\begin{eqnarray}
\label{1eqn:score_equation_I_alpha}
\dfrac{\sum\limits_{x\in\mathcal{X}} \widehat{P}(x) P_\theta(x)^{\alpha-1}s(x;\theta)}{\sum\limits_
	{x\in\mathcal{X}}\widehat{P}(x) P_\theta(x)^{\alpha-1}}= \dfrac{\sum\limits_{x\in\mathcal{X}}P_\theta(x)^\alpha s(x;\theta)}{\sum\limits_{x\in\mathcal{X}}P_\theta(x)^\alpha},
\end{eqnarray}
where $\alpha>1$. It should be noted that all the three estimating equations coincide with the usual score equation (\ref{1eqn:score_equation_mle}) when $\alpha=1$,
since $\sum_x P_\theta(x)s(x;\theta)=0$. The estimating equations (\ref{1eqn:score_equation_mle}), (\ref{1eqn:score_equation_for_D_alpha}), (\ref{1eqn:score_equation_B_alpha}),
and (\ref{1eqn:score_equation_I_alpha}) are, respectively, associated with the divergences in (\ref{1eqn:kullback-leibler}), (\ref{1defn:D_alpha_divergence}), (\ref{1defn:B_alpha_divergence}), and (\ref{1defn:I_alpha_divergence}) in a sense that will be made clear in the following.

Observe that the estimating equations (\ref{1eqn:score_equation_mle}), (\ref{1eqn:score_equation_for_D_alpha}), (\ref{1eqn:score_equation_B_alpha}),
and (\ref{1eqn:score_equation_I_alpha}) are implications of the first order optimality condition of,
respectively, the usual log-likelihood function and the following modified likelihood functions,
\begin{equation}
\label{1eqn:likelihood_function_for_D_alpha}
L_1^{(\alpha)}(\theta) := \dfrac{1}{1 - \alpha}\log\Big[\sum\limits_{x\in\mathcal{X}}\widehat{P}(x)^\alpha P_\theta
(x)^{1-\alpha}\Big],
\end{equation}
\begin{equation}
\label{1eqn:likelihood_function_for_B_alpha}
L_2^{(\alpha)}(\theta) := \dfrac{1}{n}\sum\limits_{j=1}^n\left[\dfrac{\alpha P_\theta(X_j)^{\alpha-1}-1}{\alpha-1}\right] - \sum\limits_{x\in\mathcal{X}}P_\theta(x)^\alpha,
\end{equation}
and
\begin{equation}
\label{1eqn:likelihood_function_for_I_alpha}
L_3^{(\alpha)}(\theta) := \dfrac{\alpha}{\alpha-1}\log\bigg[\dfrac{1}{n}\sum\limits_{j=1}^n P_\theta(X_j)^{\alpha-1}\bigg] - \log\Big[\sum\limits_{x\in\mathcal{X}}P_\theta(x)^\alpha\Big].
\end{equation}
Although the above likelihood functions (\ref{1eqn:likelihood_function_for_D_alpha}), (\ref{1eqn:likelihood_function_for_B_alpha}) and (\ref{1eqn:likelihood_function_for_I_alpha}) are not defined for $\alpha=1$, it can be shown that they all coincide with the following usual log-likelihood function as $\alpha\to 1$:
\begin{eqnarray}
\label{1eqn:log_likelihood_function}
L(\theta) := \frac{1}{n}\sum\limits_{j=1}^n\log P_\theta(X_j).
\end{eqnarray}
Moreover, it is easy to see that the $P_\theta$ that maximizes the likelihood function in (\ref{1eqn:likelihood_function_for_D_alpha}), (\ref{1eqn:likelihood_function_for_B_alpha}), (\ref{1eqn:likelihood_function_for_I_alpha}) or (\ref{1eqn:log_likelihood_function}) is same as the one that minimizes $D_\alpha(\widehat{P}, P_\theta)$, $B_\alpha(\widehat{P}, P_\theta)$, $\mathscr{I}_\alpha(\widehat{P}, P_\theta)$ or $I(\widehat{P}, P_\theta)$ respectively. Thus, for MLE or ``robustified MLE", one
needs to solve
\begin{eqnarray}
\label{1eqn:reverse_projection}
 \min_{P_\theta\in \Pi} D(\widehat{P},P_\theta),
\end{eqnarray}   
where $D$ is either $I$, $D_\alpha$, $B_\alpha$ or $\mathscr{I}_\alpha$. The minimizing probability measure $P_{\theta^*}$ (if exists and unique) is known as the {\em reverse $D$-projection} of $\widehat{P}$ on $\Pi$.

A ``dual'' minimization problem is the so-called {\em forward projection} problem, where the minimization is over the first argument of the divergence function. Given $\mathbb{C}\subset \mathcal{P}$, and $Q\in\mathcal{P}$, any $P^*\in \mathbb{C}$ that attains 
\begin{eqnarray}
\label{1eqn:forward_projection}
 \min_{P\in \mathbb{C}} D(P,Q)
\end{eqnarray}
is called a forward $D$-projection of $Q$ on $\mathbb{C}$. Forward projection is usually on a convex set or on an $\alpha$-convex set of probability measures. Forward projection on a convex set is motivated by the well-known {\em maximum entropy principle} of statistical physics \cite{Jaynes82J}. Motivation for forward projection on $\alpha$-convex set comes from the so-called {\em non-extensive} statistical physics \cite{Tsallis88J, TsallisMP98J, KumarS15J1}. Forward I-projection on convex sets was extensively studied by Csisz\'ar \cite{Csiszar75J, Csiszar84J, Csiszar95J2}, Csisz\'ar and Mat\'u\v{s} \cite{CsiszarM12J}, \cite{CsiszarM03J}, Csisz\'ar and Shields \cite{CsiszarS04B}, and Csisz\'ar and Tusn\'ady \cite{CsiszarT84J}.

Csisz\'ar and Shields showed that the reverse I-projection on an {\em exponential family} is same as the forward I-projection on a {\em linear family}\footnote{See Definition \ref{1defn:linear_family}.}, which, in turn, is a solution to a system of linear equations \cite[Th. 3.3]{CsiszarS04B}. We call such a result a {\em projection theorem} for the associated divergence. Projection theorem for I-divergence was due to an ``orthogonal'' relationship between the exponential family and the linear family. Projection theorem for $\mathscr{I}_\alpha$-divergence was established by Kumar and Sundaresan where the so-called {\em $\alpha$-power-law family} plays the role of exponential family \cite[Th.~18 and Th.~21]{KumarS15J2}. Projection theorem for $D_\alpha$-divergence was established by Kumar and Sason where a variant of the $\alpha$-power-law family, called an {\em $\alpha$-exponential family}, plays the role of exponential family and the so-called {\em $\alpha$-linear family} plays the role of linear family \cite[Th.~6]{KumarS16J}. Projection theorem for the more general class of Bregman divergences, where $B_\alpha$ is a subclass, was established by Csisz\'ar and Mat\'u\v{s} \cite{CsiszarM12J} using techniques from convex analysis. One of our goals in this paper is to derive the projection theorem for the $B_\alpha$-divergence using elementary tools. We also identify the parametric family of probability measures associated with the projection theorem of $B_\alpha$-divergence, which is yet another generalization of the exponential family. In all these projection theorems, the {\em Pythagorean theorem} of the respective divergence plays a key role. Thus projection theorems enable us to find the estimator (whether MLE or robustified MLE) as a forward projection if the estimation is done on a specific parametric family. While for MLE the required family is the exponential family, for robustified MLE, it is one of the generalized exponential families.

\vspace*{0.2cm}
Our contributions in this paper are the following.

\begin{itemize}
	
	\item[(a)] We show that the generalized Hellinger estimating equation (\ref{1eqn:score_equation_for_D_alpha}) and the estimating equation of Jones et al. (\ref{1eqn:score_equation_I_alpha}) are equivalent under a transformation.
	\vspace*{0.2cm}
	
	\item[(b)] We show the equivalence of the two methods of solving the estimation problems: the one by the projection theorem and the other by directly solving the corresponding estimating equation.
	\vspace*{0.2cm}
	
	\item[(c)] We show that the statistics of the data that have bearing on the projection theorem (that is, projection equation) are the sufficient statistics of the underlying statistical model with respect to the associated likelihood function.
	\vspace*{0.2cm}
	
	\item[(d)] We derive the projection theorem for the $B_\alpha$-divergence.

	
\end{itemize}

Rest of the paper is organised as follows. The contributions mentioned in (a) and (b) are covered in section \ref{1sec:Relationship_Between_Projection_Theorems_and_Estimating_equations}. The contribution mentioned in (c) is covered in section \ref{1sec:Projection_theorems_sufficiency_principle}. The contribution in (d) is given in Appendix \ref{1sec:the_projection_geometry_of_density_power_divergence}. The paper ends with a summary in section \ref{1sec:summary}.

\section{Relationship Between Projection Theorems and Estimating Equations}
\label{1sec:Relationship_Between_Projection_Theorems_and_Estimating_equations}
We now show that the solutions of the estimation problems obtained using the projection theorems and those obtained by directly solving the estimating equations are the same. Recall that, projection theorem of a divergence enables us to find the reverse projection in terms of a forward projection if the former is done on the statistical model associated with the divergence. The statistical models associated with the $I$, $B_\alpha$, $\mathscr{I}_\alpha$ and $D_\alpha$ divergences are, respectively, exponential family, non-normalized $\alpha$-power-law family, $\alpha$-power-law family, and $\alpha$-exponential family which we shall define now. Indeed, for $I$, $B_\alpha$, $\mathscr{I}_\alpha$ and $D_\alpha$ divergences, the reverse projection (if exists) must satisfy an equation that we call a {\em projection equation}. We shall first list down the statistical models along with the projection equations associated with these divergences.

\begin{enumerate}
	\item {\em Projection equation for $I$-divergence on exponential family}:
	\vspace*{0.2cm}
	
	The {\em exponential family}, $\mathcal{E} := \mathcal{E}(Q, f, \Theta)$, characterized by a probability
	measure $Q\in\mathcal{P}$, $k$ real valued functions $f_i, i=1,\dots, k$ on $\mathcal{X}$, and parameter space $\Theta\subset\mathbb{R}^k$, is given by $\mathcal{E} = \{ P_\theta : \theta\in\Theta\}\subset\mathcal{P}$,
	where
	\begin{equation}
	\label{1eqn:exponential_family_formula}
	P_\theta(x) = Z(\theta)^{-1}\exp\big[\log(Q(x)) + \theta^Tf(x)\big]\quad \text{for } x\in\mathcal{X},
	\end{equation}
	$Z(\theta)$ is a normalizing constant that makes $P_\theta$ a probability measure, and $f = (f_1,\dots,f_k)^T$.
	
	The reverse $I$-projection of $\widehat{P}$ on $\mathcal{E}$, if exists, is a unique solution of
	\begin{eqnarray}
	\label{1eqn:projection_equation_of_KL_div}
	\mathbb{E}_\theta[f(X)] = \bar{f},
	\end{eqnarray}
	where $\bar{f} := (\bar{f_1},\dots, \bar{f_k})^T$, $\bar{f_i} := \frac{1}{n}\sum_{j=1}^{n}f_i(X_j)$ for $i = 1,\dots, k$, and $\mathbb{E}_\theta$ denotes expected value with respect to $P_\theta$. The above result was due to \cite[Th.~3.3]{CsiszarS04B}. We call (\ref{1eqn:projection_equation_of_KL_div}) the {\em projection equation} for $I$-divergence on exponential family.
	\vspace*{0.2cm}
	
	\item {\em Projection equation for $B_\alpha$ on non-normalized $\alpha$-power-law family}:
	\vspace*{0.2cm}
	
	The {\em non-normalized $\alpha$-power-law family}, $\mathbb{B}^{(\alpha)} := \mathbb{B}^{(\alpha)}(Q, f, \Theta)\subset\mathcal{P}$, characterized by a $Q\in\mathcal{P}$, $k$ real valued functions $f_i, i=1,\dots, k$ on $\mathcal{X}$, and parameter space $\Theta\subset\mathbb{R}^k$, is defined as in Definition \ref{1defn:non_normalized_alpha_powerlaw_family}.
	
	The reverse $B_\alpha$-projection of $\widehat{P}$ on $\mathbb{B}^{(\alpha)}$, if exists, is a unique solution of
	\begin{equation*}
	\mathbb{E}_\theta[f(X)] = \bar{f},
	\end{equation*}
	by Theorem \ref{1thm:orthogonality2}. (Notice that the projection equations for $I$ and $B_\alpha$-divergences are the same.)
	\vspace*{0.2cm}
	
	\item {\em Projection equation for $\mathscr{I}_\alpha$ on $\alpha$-power-law family}:
	\vspace*{0.2cm}
	
	The {\em $\alpha$-power-law family}, $\mathbb{M}^{(\alpha)} := \mathbb{M}^{(\alpha)}(Q, f, \Theta)$\footnote{In the continuous case, each Student-t distribution can be seen as point-wise limit of a sequence of distributions from $\mathbb{M}^{(\alpha)}$. (c.f. \cite[Rem.~13]{KumarS15J1}.)}, characterized by a $Q\in\mathcal{P}$ and $k$ real valued functions $f_i, i = 1,\dots, k$ on $\mathcal{X}$, and parameter space $\Theta\subset\mathbb{R}^k$, is given by	$\mathbb{M}^{(\alpha)} = \{ P_\theta : \theta\in\Theta\}\subset\mathcal{P}$,	where
	\begin{eqnarray}
	\label{1eqn:power_law_family_formula}
	P_\theta(x) = Z(\theta)^{-1} \big[Q(x)^{\alpha-1} + (1-\alpha)\theta^Tf(x)\big]^{\frac{1}{\alpha-1}}\quad\text{for }~x\in\mathcal{X},
	\end{eqnarray}
	and $Z(\theta)$ is the normalizing constant (c.f. \cite[Defn.~8]{KumarS15J2}).
	\vspace*{0.2cm}
	
	The reverse $\mathscr{I}_\alpha$-projection of $\widehat{P}$ on $\mathbb{M}^{(\alpha)}$, if exists, is a unique solution of
	\begin{eqnarray}
	\label{1eqn:projection_equation_of_I_alpha_div}
	\mathbb{E}_\theta\big[f_i(X)\big] = \dfrac{\mathbb{E}_\theta\big[Q(X)^{\alpha-1}\big]}{\overline{{Q}^{\alpha-1}}}\bar{f_i},
	\end{eqnarray}
	for $i = 1,\dots, k$, where $\overline{{Q}^{\alpha-1}} := \frac{1}{n}\sum_{j=1}^n Q(X_j)^{\alpha-1}$ \cite[Th.~18 and Th.~21]{KumarS15J2}\footnote{{\label{1ftnote:linear_independence}}Notice that this result in \cite{KumarS15J2} is true under the assumption that the associated linear family $\mathbb{L}$ determined by the functions $f_1,\dots,f_k$ is non-empty.}.
	\vspace*{0.2cm}
	
	\item {\em Projection equation for $D_\alpha$ on $\alpha$-exponential family}:
	\vspace*{0.2cm}
	
	The {\em $\alpha$-exponential family}, $\mathscr{E}_\alpha := \mathscr{E}_\alpha(Q, f, \Theta)$, characterized by a $Q\in\mathcal{P}$ and $k$ real valued functions $f_i, i = 1,\dots, k$ on $\mathcal{X}$, and parameter space $\Theta\subset\mathbb{R}^k$, is given by $\mathscr{E}_\alpha = \{ P_\theta : \theta\in\Theta\}\subset\mathcal{P}$,
	where
	\begin{eqnarray}
	\label{1eqn:E_alpha_family_formula}
	P_\theta(x) = Z(\theta)^{-1} \big[Q(x)^{1-\alpha} + (1-\alpha)\theta^Tf(x)\big]^{\frac{1}{1-\alpha}}\quad\text{for }~x\in\mathcal{X},
	\end{eqnarray}
	and $Z(\theta)$ is the normalizing constant (c.f. \cite[Eq.~(62)]{KumarS16J}.
	\vspace*{0.2cm}
	
	The reverse $D_\alpha$-projection of $\widehat{P}$ on $\mathscr{E}_{\alpha}$, if exists, is a unique solution of
	\begin{eqnarray}
	\label{1eqn:projection_equation_of_D_alpha_div}
	\sum\limits_{x\in\mathcal{X}}P_\theta(x)^\alpha f_i(x) = \dfrac{\sum\limits_{x\in\mathcal{X}}P_\theta(x)^\alpha Q(x)^{1-\alpha}}{\sum\limits_{x\in\mathcal{X}}\widehat{P}(x)^\alpha Q(x)^{1-\alpha}}\sum\limits_{x\in\mathcal{X}}\widehat{P}(x)^\alpha f_i(x),
	\end{eqnarray}
	for $i = 1,\dots, k$ (\cite[Th.~6]{KumarS16J}).
\end{enumerate}
(\ref{1eqn:projection_equation_of_D_alpha_div}) can be re-written as
\begin{eqnarray}
\label{1eqn:projection_equation_of_D_alpha_div1}
\mathbb{E}_{\theta^{(\alpha)}}\big[f_i(X)\big] = \dfrac{\mathbb{E}_{\theta^{(\alpha)}}\big[Q(X)^{1-\alpha}\big]}{\overline{{Q}^{1-\alpha}}^{(\alpha)}}\overline{f_i}^{(\alpha)},
\end{eqnarray}
where $\mathbb{E}_{\theta^{(\alpha)}}$ denotes expectation with respect to $P_{\theta}^{(\alpha)}$; $\overline{{Q}^{1-\alpha}}^{(\alpha)}$ and $\overline{f_i}^{(\alpha)}$ are, respectively, averages of $Q(\cdot)^{1-\alpha}$ and $f_i(\cdot)$, with respect to $\widehat{P}^{(\alpha)}$,  where $P^{(\alpha)}(x) := P(x)^{\alpha}/\sum_y P(y)^{\alpha}$ is the {\em $\alpha$-scaled measure}\footnote{Sometimes referred as {\em escort measure} among the physicists.} associated with $P$.

Notice that while the projection equations for $I$ and $B_{\alpha}$ depend on the sample only through $\bar{f}$, the projection equations for $\mathscr{I}_{\alpha}$ and $D_\alpha$, respectively, depend on $\bar{f}/{\overline{{Q}^{\alpha-1}}}$ and $\overline{f}^{(\alpha)}/{\overline{{Q}^{1-\alpha}}^{(\alpha)}}$. This fact will be further explored in terms of sufficiency principle in section \ref{1sec:Projection_theorems_sufficiency_principle}.  Also notice that (\ref{1eqn:projection_equation_of_I_alpha_div}) and (\ref{1eqn:projection_equation_of_D_alpha_div1}) are related by the transformations $P\leftrightarrow P^{(\alpha)}$ and $\alpha\leftrightarrow 1/\alpha$. The following lemmas explore this connection further. While lemma \ref{1lem:equivalence_of_E_alpha_and_I_alpha} establishes the connection between $\mathbb{M}^{(\alpha)}$ and $\mathscr{E}_\alpha$, lemma \ref{1lem:connection_of_estimating_equation_and_projection_equation} establishes the connection between the estimation problem based on $D_\alpha$-divergence on $\mathscr{E}_\alpha$ and the estimation problem based on $\mathscr{I}_\alpha$-divergence on $\mathbb{M}^{(\alpha)}$.
Lemma \ref{1lem:equivalence_of_E_alpha_and_I_alpha} is due to Karthik and Sundaresan \cite[Th.~2]{KarthikS17J}, where only the reverse implication was proved though. 
\begin{lemma}
	\label{1lem:equivalence_of_E_alpha_and_I_alpha}
	The map $P\mapsto P^{(\alpha)}$ establishes a one-to-one correspondence between the parametric families
	$\mathscr{E}_\alpha (Q, f, \Theta)$ and $\mathbb{M}^{(1/\alpha)}(Q^{(\alpha)}, f, \Theta')$,
	where 
	\begin{eqnarray}
	\label{1eqn:theta'}
	\Theta' = \Big\lbrace \theta'=(\theta_1',\dots, \theta_k'):\theta'_i = \tfrac{(- \alpha)\theta_i}{\| Q \| ^{1 - \alpha}}, i=1,\dots,k,~\theta\in \Theta\Big\rbrace,
	\end{eqnarray}
	and $\|Q\| := \big[\sum_x Q(x)^\alpha\big]^{1/\alpha}$. (We suppress the dependence of $\|Q\|$ on $\alpha$ for notational convenience.)
\end{lemma}

\begin{proof}
	For any $P_\theta\in \mathscr{E}_\alpha (Q, f, \Theta)$, we have from (\ref{1eqn:E_alpha_family_formula}), for $x\in \mathcal{X}$,
	\begin{align}
	\label{1eqn:E_alpha_to_M_alpha}
	P^{(\alpha)}_\theta(x) 
	& = \tfrac{Z(\theta)^{ - \alpha}}{\|P_\theta\|^\alpha} \Big[Q(x)^{1 - \alpha} + (1 - \alpha) \theta^Tf(x)\Big]^{\frac{\alpha}{1 - \alpha}}\nonumber\\
	& = \tfrac{Z(\theta)^{ - \alpha}\|Q\|^{\alpha}}{\|P_\theta\|^\alpha} \Big[ \Big\lbrace {\tfrac{Q(x)^\alpha}{\|Q\|^\alpha}}\Big\rbrace^{\frac{1 - \alpha}{\alpha}} + 
	(1-\tfrac{1}{\alpha}) \frac{(- \alpha)}{\| Q \| ^{1-\alpha}}
	\theta^T f(x)\Big]^{\frac{1}{\frac{1}{\alpha}-1}}\nonumber\\
	& =   \tfrac{Z(\theta)^{ - \alpha}\|Q\|^{\alpha}}{\|P_\theta\|^\alpha} \Big[ \big\lbrace Q^{(\alpha)}(x)\big\rbrace^{\frac{1}{\alpha}-1}
	+ (1-\tfrac{1}{\alpha}) {\theta'}^T f(x)\Big]^{\frac{1}{\frac{1}{\alpha}-1}}.
	\end{align}
	Hence $P_\theta^{(\alpha)}\in \mathbb{M}^{(1/\alpha)}(Q^{(\alpha)}, f, \Theta')$. So, the mapping is well-defined. The map is clearly one-one, since it is easy to see that, if $P_{\theta}^{(\alpha)} = P_{\eta}^{(\alpha)}$ for some $\theta, \eta\in\Theta$, then $P_{\theta} = P_{\eta}$.
	To verify it is onto, let $P\in \mathbb{M}^{(1/\alpha)}(Q^{(\alpha)}, f, \Theta')$
	be arbitrary. Then
	\begin{align*}
	P(x) 
	& = Z(\theta')^{-1}\Big[\big\lbrace Q^{(\alpha)}(x)\big\rbrace^{\frac{1}{\alpha} - 1} + 
	(1 - \tfrac{1}{\alpha}) {\theta'}^Tf(x)\Big]^{\frac{1}{\frac{1}{\alpha}-1}}\\
	& =  Z(\theta')^{-1}\Big[\Big\lbrace \tfrac{Q(x)}{\|Q\|}\Big\rbrace^{1-\alpha} +
	\tfrac{(1-\alpha)}{\| Q \| ^{1-\alpha}} \theta^Tf(x)\Big]^{\frac{\alpha}{1-\alpha}}\\
	& = \tfrac{Z(\theta')^{-1}}{\|Q\|^\alpha} \big[Q(x)^{1-\alpha} + (1-\alpha) \theta^Tf(x)\big]^{\frac{\alpha}{1-\alpha}}.
	\end{align*}
	This implies that
	\begin{equation*}
	P(x)^{1/\alpha} = \tfrac{Z(\theta')^{{-1/\alpha}}}{\|Q\|} \big[Q(x)^{1-\alpha} + (1-\alpha)\theta^Tf(x)\big]^{\frac{1}{1-\alpha}},
	\end{equation*}
	and hence
	\begin{align*}
	P^{(1/\alpha)}(x)
	& = \tfrac{Z(\theta')^{-1/\alpha}}{\|Q\| \sum_y P(y)^{1/\alpha}} \big[Q(x)^{1-\alpha} + (1-\alpha)\theta^Tf(x)\big]^{\frac{1}{1-\alpha}}.
	\end{align*}
	Hence $P^{(1/\alpha)}\in \mathscr{E}_\alpha (Q, f, \Theta)$ and so $P^{(1/\alpha)} = P_\theta$ for some $\theta\in\Theta$. It is now easy to show that $P_\theta^{(\alpha)} = P$. Thus, for any $P\in \mathbb{M}^{(\tiny{1/\alpha})}(Q^{(\alpha)}, f, \Theta')$, there exists $P_\theta\in \mathscr{E}_\alpha (Q, f, \Theta)$
	such that $P^{(\alpha)}_\theta = P$. Hence the mapping is onto.
	
\end{proof}

\begin{lemma}
	\label{1lem:connection_of_estimating_equation_and_projection_equation}
	The following statements hold.
	\begin{enumerate}
		\item[(a)] Solving the estimating equation (\ref{1eqn:score_equation_for_D_alpha})
		for $P_\theta\in \mathscr{E}_\alpha(Q, f, \Theta)$
		is equivalent to solving the estimating equation (\ref{1eqn:score_equation_I_alpha}) for the $\alpha$-scaled measure $P_\theta^{(\alpha)} \in \mathbb{M}^{(\tiny{1/\alpha})}(Q^{(\alpha)}, f, \Theta')$, where $\Theta'$ is as in (\ref{1eqn:theta'}).
		
		\item[(b)] Solving the projection equation (\ref{1eqn:projection_equation_of_D_alpha_div1}) for 
		$P_\theta\in \mathscr{E}_\alpha(Q, f, \Theta)$ is equivalent to solving the projection equation 
		(\ref{1eqn:projection_equation_of_I_alpha_div}) for the $\alpha$-scaled measure	$P_\theta^{(\alpha)}\in \mathbb{M}^{({1/\alpha})}(Q^{(\alpha)},	f, \Theta')$,
		 where $\Theta'$ is as in (\ref{1eqn:theta'}).
	\end{enumerate}
\end{lemma}
\begin{proof}
	\begin{itemize}
		\item[(a)] The estimating equation in (\ref{1eqn:score_equation_for_D_alpha}) can be re-written as
		\begin{eqnarray}
		\label{1eqn:modified_estimating_equation_for_D_alpha}
		\dfrac{\sum\limits_{x\in\mathcal{X}}\widehat{P}(x)^\alpha P_\theta(x)^{1-\alpha}s(x;\theta)}{\sum\limits_{x\in\mathcal{X}}\widehat{P}(x)^\alpha P_\theta(x)^{1-\alpha}} =
		{\sum\limits_{x\in\mathcal{X}}P_\theta(x)s(x;\theta)},
		\end{eqnarray}
		since $\sum\limits_{x\in\mathcal{X}}P_\theta(x)s(x;\theta) = 0$. This can further be re-written as
		\begin{eqnarray}
		\label{1eqn:estimating _equation_of_D_alpha_after_transformation}
		\dfrac{\sum\limits_{x\in\mathcal{X}}\widehat{P}^{(\alpha)}(x) \big[ P^{(\alpha)}_\theta(x)\big]
			^{\tfrac{1}{\alpha}-1}s(x;\theta)}{\sum\limits_{x\in\mathcal{X}}\widehat{P}^{(\alpha)} (x)
			\big[P_\theta^{(\alpha)}(x)\big] ^{\tfrac{1}{\alpha}-1}}  =
		\dfrac{\sum\limits_{x\in\mathcal{X}}\big[ P^{(\alpha)}_\theta(x)\big]
			^{\tfrac{1}{\alpha}} s(x;\theta)}
		{\sum\limits_{x\in\mathcal{X}}\big[ P^{(\alpha)}_\theta(x)\big]
			^{\tfrac{1}{\alpha}}}. 
		\end{eqnarray}
		Observe that
		\begin{align*}
		s^{(\alpha)}(x;\theta) := \nabla\log P^{(\alpha)}_\theta(x) 
		& = \nabla\log \tfrac{P_\theta(x)^\alpha}{\|P_\theta\|^\alpha}\\
		& = \nabla \big[ \log P_\theta (x)^\alpha -\log \|P_\theta\|^\alpha\big]\\
		& = \alpha \big[s(x;\theta) - \nabla\log \|P_\theta\| \big].
		\end{align*}
		Hence
		\begin{eqnarray}
		\label{1eqn:s_in_terms_of_s_alpha}
		s(x;\theta) = \tfrac{1}{\alpha} s^{(\alpha)}(x;\theta) + A(\theta),
		\end{eqnarray}
		where $A(\theta) = \nabla\log \|P_\theta\|$. Plugging (\ref{1eqn:s_in_terms_of_s_alpha}) in (\ref{1eqn:estimating _equation_of_D_alpha_after_transformation}), we get
		\begin{eqnarray}
		\label{1eqn:estimating _equation_of_D_alpha_after_transformation_final}
		\dfrac{\sum\limits_{x\in\mathcal{X}}\widehat{P}^{(\alpha)}(x) \big[ P^{(\alpha)}_\theta(x)\big]
			^{\tfrac{1}{\alpha}-1}s^{(\alpha)}(x;\theta)}{\sum\limits_{x\in\mathcal{X}}\widehat{P}^{(\alpha)} (x)
			\big[P_\theta^{(\alpha)}(x)\big] ^{\tfrac{1}{\alpha}-1}}  =
		\dfrac{\sum\limits_{x\in\mathcal{X}}\big[ P^{(\alpha)}_\theta(x)\big]
			^{\tfrac{1}{\alpha}} s^{(\alpha)}(x;\theta)}
		{\sum\limits_{x\in\mathcal{X}}\big[ P^{(\alpha)}_\theta(x)\big]
			^{\tfrac{1}{\alpha}}}. 
		\end{eqnarray}
		This is same as (\ref{1eqn:score_equation_I_alpha}) with $\widehat{P}$, $P_\theta$, and $\alpha$, respectively, replaced by $\widehat{P}^{(\alpha)}$, $P_\theta^{(\alpha)}$, and $1/\alpha$. Thus, solving the estimating equation (\ref{1eqn:score_equation_for_D_alpha}) for $P_\theta\in \mathscr{E}_\alpha(Q, f, \Theta)$ is equivalent to solving (\ref{1eqn:score_equation_I_alpha}) for  $P_\theta^{(\alpha)} \in \mathbb{M}^{({1/\alpha})}(Q^{(\alpha)}, f, \Theta')$, by lemma \ref{1lem:equivalence_of_E_alpha_and_I_alpha}.
		\vspace*{0.2cm}
		
		\item[(b)] The projection equation (\ref{1eqn:projection_equation_of_D_alpha_div1}) can further be re-written as
		\begin{eqnarray}
		\label{1eqn:projection_equation_of_D_alpha_transformed}
		\mathbb{E}_{\theta^{(\alpha)}}\big[f_i (X)\big] = 
		\dfrac{\mathbb{E}_{\theta^{(\alpha)}}
			\Big[\big\lbrace Q^{(\alpha)}(X)\big\rbrace^{\tfrac{1}{\alpha}-1}\Big]}
		{\sum\limits_{x\in\mathcal{X}}\widehat{P}^{(\alpha)}(x) \big[Q^{(\alpha)}(x)\big]^{\tfrac{1}{\alpha}-1}}
		\sum\limits_{x\in\mathcal{X}}\widehat{P}^{(\alpha)}(x) f_i(x).
		\end{eqnarray}
		This is same as (\ref{1eqn:projection_equation_of_I_alpha_div}) with $P_\theta$, $\widehat{P}$, $Q$, and $\alpha$, respectively, replaced by $P_\theta^{(\alpha)}$, $\widehat{P}^{(\alpha)}$, $Q^{(\alpha)}$, and $1/\alpha$. Thus solving (\ref{1eqn:projection_equation_of_D_alpha_div}) for 
		$P_\theta\in \mathscr{E}_\alpha(Q, f, \Theta)$
		is equivalent to solving (\ref{1eqn:projection_equation_of_D_alpha_transformed}) for
		$P_\theta^{(\alpha)}\in \mathbb{M}^{(1/\alpha)}(Q^{(\alpha)}, 
		f, \Theta')$, by lemma \ref{1lem:equivalence_of_E_alpha_and_I_alpha}.
	\end{itemize}  
\end{proof}

We are now ready to state the main result of this section.

\begin{theorem}
	The solutions of the estimation problems based on the estimating equations (\ref{1eqn:score_equation_mle}), (\ref{1eqn:score_equation_for_D_alpha}), (\ref{1eqn:score_equation_B_alpha}),  and (\ref{1eqn:score_equation_I_alpha}) on  $\mathcal{E}$, $\mathscr{E}_\alpha$, $\mathbb{B}^{(\alpha)}$, and $\mathbb{M}^{(\alpha)}$ respectively, is same as the solutions obtained from their corresponding projection equations. 
\end{theorem}

\begin{proof}
	
	\begin{enumerate}
		\item[(a)] {\em MLE on exponential family $\mathcal{E}$}:
		\vspace{0.2cm}
		
		For $P_\theta\in\mathcal{E}$,
		\begin{equation}
		\label{1eqn:ln_exponential_probability}
		\log P_\theta(x) = -\log Z(\theta) + \log Q(x) + \theta^Tf(x).
		\end{equation}
		Hence (\ref{1eqn:score_equation_mle_in_terms_of_sample}) implies that
		\begin{equation}
		\label{1eqn:implication_of_score_equation}
		\nabla\log Z(\theta) = \frac{1}{n}\sum\limits_{j=1}^{n}f(X_j).
		\end{equation}
		Since $\mathbb{E}_\theta[\nabla\log P_\theta(X)] = 0$, (\ref{1eqn:ln_exponential_probability}) implies that
		\begin{equation}
		\label{1eqn:implication_of_expected_score_being_0}
		\nabla\log Z(\theta) = \mathbb{E}_\theta[f(X)].
		\end{equation}
		Thus, from (\ref{1eqn:implication_of_score_equation}) and (\ref{1eqn:implication_of_expected_score_being_0}), we see that the MLE must satisfy
		\begin{equation*}
		\mathbb{E}_\theta[f(X)] = \bar{f},
		\end{equation*}
		which is same as the projection equation for $I$-divergence.
		\vspace*{0.2cm}
		
		\item[(b)] {\em Robust estimation based on (\ref{1eqn:score_equation_B_alpha}) on $\mathbb{B}^{(\alpha)}$}:
		\vspace*{0.2cm}
		
		The estimating equation (\ref{1eqn:score_equation_B_alpha}) can be re-written as
		\begin{eqnarray}
		\label{1eqn:score_equation_B_alpha_calculated}
		\sum\limits_{x\in\mathcal{X}}\widehat{P}(x)P_\theta(x)^{\alpha-2}\nabla P_\theta(x) = \sum\limits_{x\in\mathcal{X}}P_\theta(x)^{\alpha-1}\nabla P_\theta(x).
		\end{eqnarray}
		For $P_\theta\in \mathbb{B}^{(\alpha)}$, from (\ref{1eqn:B_alpha_family_formula}) we have
		\begin{eqnarray*}
		\nabla P_\theta(x) & = & -\big[Q(x)^{\alpha-1} + (1-\alpha) \big\{Z(\theta) + \theta^Tf(x)\big\}\big]^{\frac{1}{\alpha-1} - 1}[\nabla Z(\theta) + f(x)]\\
		& = & -P_\theta(x)^{2-\alpha}[\nabla Z(\theta) + f(x)].
		\end{eqnarray*}
		Substituting this in (\ref{1eqn:score_equation_B_alpha_calculated}), we get
		\begin{equation*}
		\mathbb{E}_\theta[f(X)] = \bar{f},
		\end{equation*}
		which is same as the projection equation for $B_\alpha$-divergence.
		\vspace*{0.2cm}
		
		\item[(c)] {\em Robust estimation based on (\ref{1eqn:score_equation_I_alpha}) on $\mathbb{M}^{(\alpha)}$}:
		\vspace*{0.2cm}
		
		The estimating equation (\ref{1eqn:score_equation_I_alpha}) can be re-written as
		\begin{eqnarray}
		\label{1eqn:score_equation_I_alpha_calculated_2}
		\dfrac{\frac{1}{n}\sum\limits_{i=1}^n P_\theta(X_i)^{\alpha-2}\nabla P_\theta(X_i)}{\frac{1}{n}\sum\limits_{i=1}^n P_\theta(X_i)^{\alpha-1}} =
		\dfrac{\sum\limits_{x\in\mathcal{X}}P_\theta(x)^{\alpha-1}\nabla P_\theta(x)}{\sum\limits_{x\in\mathcal{X}}P_\theta(x)^{\alpha}}.
		\end{eqnarray}
		Now if $P_\theta\in \mathbb{M}^{(\alpha)}$, then from (\ref{1eqn:power_law_family_formula}) we have
		\begin{equation}
		\label{1eqn:alpha_power_law_formula}
		P_\theta(x)^{\alpha-1} = Z(\theta)^{1-\alpha}\big[ Q(x)^{\alpha-1}+(1-\alpha)\theta^Tf(x)\big]\quad \text{for } x\in \mathcal{X}.
		\end{equation}
		Differentiating (\ref{1eqn:alpha_power_law_formula}) with respect to $\theta$, we get
		\begin{eqnarray}
		\label{1eqn:derivative_of_P_theta_in_I_family}
		\lefteqn{P_\theta(x)^{\alpha-2}\nabla P_\theta(x)}\nonumber\\
		&& = -Z(\theta)^{-\alpha}\nabla Z(\theta)\big[Q(x)^{\alpha-1}+(1-\alpha)\theta^Tf(x)\big] - Z(\theta)^{1-\alpha}f(x)\nonumber\\
		&& = -Z(\theta)^{1-\alpha}\left\{ Z(\theta)^{-1}\nabla Z(\theta)\big[Q(x)^{\alpha-1}+(1-\alpha)\theta^Tf(x)\big] + f(x)\right\}.
		\end{eqnarray}
		Using (\ref{1eqn:alpha_power_law_formula}) and (\ref{1eqn:derivative_of_P_theta_in_I_family}), left-hand side of (\ref{1eqn:score_equation_I_alpha_calculated_2}) becomes
		\begin{eqnarray*}
			-\dfrac{\frac{1}{n}\sum\limits_{j=1}^n\Big\lbrace Z(\theta)^{-1}\nabla Z(\theta)\big[Q(X_j)^{\alpha-1}+(1-\alpha)
				\theta^Tf(X_j) \big] + f(X_j)\Big\rbrace}{\frac{1}{n}\sum\limits_{j=1}^n\big[Q(X_j)^{\alpha-1}+(1-\alpha)
				\theta^Tf(X_j) \big]}\\
			=- Z(\theta)^{-1}\nabla Z(\theta) - \dfrac{\bar{f}}{\overline{Q^{\alpha-1}}+(1-\alpha)\theta^T\bar{f}}.
		\end{eqnarray*}
		Similarly, using (\ref{1eqn:alpha_power_law_formula}) and (\ref{1eqn:derivative_of_P_theta_in_I_family}), one can show that the right-hand side of (\ref{1eqn:score_equation_I_alpha_calculated_2}) is
		\begin{eqnarray*}
			-Z(\theta)^{-1}\nabla Z(\theta) - \dfrac{\sum\limits_{x\in\mathcal{X}}f(x)\big[Q(x)^{\alpha-1}+(1-\alpha)\theta^Tf(x) \big]^{\frac{1}{\alpha-1}}}
			{\sum\limits_{x\in\mathcal{X}}\big[Q(x)^{\alpha-1}+(1-\alpha)
				\theta^Tf(x) \big]^{\frac{\alpha}{\alpha-1}}}.
		\end{eqnarray*}
		Hence (\ref{1eqn:score_equation_I_alpha_calculated_2}) is same as
		\begin{equation*}
		\dfrac{\bar{f}}{\overline{Q^{\alpha-1}}+(1-\alpha)\theta^T\bar{f}} = \dfrac{\sum\limits_{x\in\mathcal{X}}f(x)\big[Q(x)^{\alpha-1}+(1-\alpha)\theta^Tf(x) \big]^{\frac{1}{\alpha-1}}}
		{\sum\limits_{x\in\mathcal{X}}\big[Q(x)^{\alpha-1}+(1-\alpha)
			\theta^Tf(x) \big]^{\frac{\alpha}{\alpha-1}}}.
		\end{equation*} 
		Using (\ref{1eqn:alpha_power_law_formula}) the above can be re-written as
		\begin{eqnarray}
		\label{1eqn:equation_for_estimator_from_estimating_equation}
		\dfrac{\bar{f}}{\overline{Q^{\alpha-1}} + (1-\alpha)\theta^T\bar{f}} = \dfrac{\mathbb{E}_\theta [f(X)]}{\mathbb{E}_\theta [Q(X)^{\alpha-1}+(1-\alpha)
			\theta^Tf(X)]}.
		\end{eqnarray}
		Hence the estimator $P_\theta$ must satisfy the above equation.
		\vspace*{0.2cm}
		
		On the other hand, the projection equation (\ref{1eqn:projection_equation_of_I_alpha_div}) can be re-written as
		\begin{equation*}
		\mathbb{E}_\theta\big[f_i(X)\big]=\dfrac{\bar{f_i}}{\overline{Q^{\alpha-1}}}
		\mathbb{E}_\theta\big[Q(X)^{\alpha-1}\big],\quad \text{ for } i = 1,\dots,k.
		\end{equation*}
		This implies
		\begin{equation*}
		\theta^T\mathbb{E}_\theta\big[f(X)\big] = \dfrac{\theta^T\bar{f}}{\overline{Q^{\alpha-1}}}
		\mathbb{E}_\theta\big[Q(X)^{\alpha-1}\big].
		\end{equation*}
		The above two equations can be combined to
		\begin{equation*}
		\dfrac{\bar{f}_i}{\mathbb{E}_\theta\big[f_i(X)\big]} = 
		\dfrac{\overline{Q^{\alpha-1}}}{\mathbb{E}_\theta\big[Q(X)^{\alpha-1}\big]}=
		\dfrac{(1-\alpha)\theta^T\bar{f}}{(1-\alpha)\theta^T\mathbb{E}_\theta\big[f(X)\big]}, \quad\text{ for } i = 1,\dots,k.
		\end{equation*}
		This implies that the reverse projection $P_\theta$ must satisfy
		\begin{equation}
		\label{1eqn:estimating_equation_I_alpha_simplified}
		\dfrac{\bar{f}_i}{\mathbb{E}_\theta\big[f_i(X)\big]}=
		\dfrac{\overline{Q^{\alpha-1}}+(1-\alpha)\theta^T\bar{f}}{\mathbb{E}_\theta\big[
			Q(X)^{\alpha-1}+(1-\alpha)\theta^Tf(X)\big]}, \quad\text{ for } i = 1,\dots,k.
		\end{equation}
		This is same as (\ref{1eqn:equation_for_estimator_from_estimating_equation}) derived from the estimating equation. Hence any solution of the projection equation (\ref{1eqn:projection_equation_of_I_alpha_div}) is a solution of the estimating equation (\ref{1eqn:equation_for_estimator_from_estimating_equation}). Hence, it suffices to show that, if (\ref{1eqn:equation_for_estimator_from_estimating_equation}) has a solution, it is unique. To this end, we define $\Phi := (\phi_1,\dots,\phi_k)$, where
		\begin{equation}
		\label{1eqn:expression_of_phi_i_1}
		\phi_i(\theta) := \dfrac{\mathbb{E}_\theta\big[f_i(X)\big]\big[\overline{Q^{\alpha-1}}+(1-\alpha)\theta^T\bar{f}\big]}{\mathbb{E}_\theta\big[Q(X)^{\alpha-1}+(1-\alpha)\theta^T f(X)\big]} = \frac{\mathbb{E}_\theta [f_i(X)]\overline{P_\theta^{\alpha-1}}}{\|P_\theta\|^\alpha}, \quad\text{ for } i = 1,\dots,k,
		\end{equation}
		where the second equality follows from (\ref{1eqn:alpha_power_law_formula}). Then (\ref{1eqn:equation_for_estimator_from_estimating_equation}) reduces to solving
		\begin{equation}
		\label{1eqn:estimating_equation_I_alpha_final}
		\Phi(\theta) = \bar{f}
		\end{equation}
		for $\theta$. Thus it remains to show that the mapping $\theta\mapsto \Phi(\theta)$ is one-one. This is indeed the case. The proof is given in Appendix \ref{1sec:phi_is_oneone}.
		
		\vspace{0.2cm}
		
		\item[(d)] {\em Robust estimation based on (\ref{1eqn:score_equation_for_D_alpha}) on $\mathscr{E}_{\alpha}$}:
		\vspace{0.2cm}
		
		This clearly follows from (c) in view of lemmas \ref{1lem:equivalence_of_E_alpha_and_I_alpha}
		and \ref{1lem:connection_of_estimating_equation_and_projection_equation}.
	\end{enumerate}
\end{proof}
\section{Projection Theorems and the Principle of Sufficiency}
\label{1sec:Projection_theorems_sufficiency_principle}
In the previous section we studied four estimation problems respectively on four statistical models from the perspectives of projection theorems and estimating equations. We observed that the estimating equations or the projection equations arising from these estimation problems depended on the given sample only through some specific statistics of the sample. In this section we explore this observation further from the principle of sufficiency. The following result, known as {\em factorization theorem} \cite[Th.~6.2.6]{CasellaB02B}, tells us a way to identify the sufficient statistics from the log-likelihood function.

Factorization Theorem: Let $X_1,\dots, X_n$ be an i.i.d. random sample drawn according to some model $P_\theta, \theta\in\Theta$. A statistic $T(X_1,\dots,X_n)$ is a sufficient statistic for $\theta$ if and only if there exists functions $g(\theta, T(X_1,\dots,X_n))$ and $h(X_1,\dots,X_n)$ such that the log-likelihood function $L(\theta)$ in (\ref{1eqn:log_likelihood_function}) can be written as
\begin{eqnarray}
\label{1eqn:criterion_for_sufficient_statistic}
L(\theta) = g(\theta, T(X_1,\dots,X_n)) + h(X_1,\dots,X_n),
\end{eqnarray}
for all sample points $(X_1,\dots,X_n)$ and for all $\theta\in\Theta$.

However, when the data set is contaminated, we saw in section \ref{1sec:introduction} that one should use some modified likelihood
function for inference instead of the usual log-likelihood function. The quantities in (\ref{1eqn:likelihood_function_for_D_alpha}), (\ref{1eqn:likelihood_function_for_B_alpha})
and (\ref{1eqn:likelihood_function_for_I_alpha}) are some examples of modified likelihood functions. Suppose that $L_G$ is one of the modified likelihood functions in (\ref{1eqn:likelihood_function_for_D_alpha}) - (\ref{1eqn:likelihood_function_for_I_alpha}). The robust estimator of $\theta$ is given by
\begin{eqnarray}
\theta_E := \displaystyle\arg\max_\theta L_G(\theta).
\end{eqnarray}
Now, if we can analogously write the modified likelihood function $L_G$ as in (\ref{1eqn:criterion_for_sufficient_statistic}), that is, if
\begin{eqnarray}
L_G(\theta) = g (\theta, T(X_1,\dots,X_n)) + h(X_1,\dots, X_n),
\end{eqnarray}
for some functions $g$ and $h$, then we have
\begin{align*}
\theta_E & = \displaystyle\arg\max_\theta L_G(\theta)\\
& = \displaystyle\arg\max_\theta [g (\theta, T(X_1,\dots,X_n)) + h(X_1,\dots, X_n)]\\
& = \displaystyle\arg\max_\theta g (\theta, T(X_1,\dots,X_n)).
\end{align*}
This means that the robustified MLE depends on the sample only through the function $T(\cdot)$. Thus it is reasonable to call such $T(\cdot)$ a sufficient statistics for $\theta$ with respect to the modified likelihood function. In view of this, in the following theorem, we
find the sufficient statistics for the parameters of each of the families $\mathcal{E}$, $\mathbb{B}^{(\alpha)}$, $\mathbb{M}^{(\alpha)}$ and $\mathscr{E}_\alpha$ when the appropriate likelihood function is used.
\begin{theorem}
	\label{1thm:sufficient_statistics_for_all}
	Let $X_1,\dots, X_n$ be an i.i.d. random sample drawn according to one of the statistical models $\mathcal{E}$, $\mathbb{B}^{(\alpha)}$, $\mathbb{M}^{(\alpha)}$ or $\mathscr{E}_\alpha$. Then the following statements hold.
	\begin{itemize}
		\item[(a)] For $\mathcal{E}$, $T_1(X_1,\dots,X_n) = \bar{f}$ is a sufficient statistic for $\theta$ with respect to the usual log-likelihood function $L(\theta)$ in (\ref{1eqn:log_likelihood_function}).
		
		\item[(b)] For $\mathbb{B}^{(\alpha)}$, $T_2(X_1\dots,X_n) = \bar{f}$ is a sufficient statistic for $\theta$ with respect to the likelihood function $L_2^{(\alpha)}(\theta)$ in (\ref{1eqn:likelihood_function_for_B_alpha}).
		
		\item[(c)] For $\mathbb{M}^{(\alpha)}$, $T_3(X_1\dots,X_n) = {\bar{f}}/{\overline{Q^{\alpha-1}}}$ is a sufficient statistic for $\theta$ with respect to the likelihood function $L_3^{(\alpha)}(\theta)$ in (\ref{1eqn:likelihood_function_for_I_alpha}).
		
		\item[(d)] For $\mathscr{E}_{\alpha}$, $T_4(X_1\dots,X_n) = {{\overline{f}}^{(\alpha)}}/{{\overline{Q^{1-\alpha}}}^{(\alpha)}}$ is a sufficient statistic for $\theta$ with respect to the likelihood function $L_1^{(\alpha)}(\theta)$ in (\ref{1eqn:likelihood_function_for_D_alpha}).
	\end{itemize}
\end{theorem}
\begin{proof}
	\begin{itemize}
		
		\item[(a)] See \cite[Th.~6.2.10]{CasellaB02B}.
		\vspace*{0.2cm}
		
		\item[(b)]	For $P_\theta\in \mathbb{B}^{(\alpha)}$, using (\ref{1eqn:B_alpha_family_formula}), the likelihood function in (\ref{1eqn:likelihood_function_for_B_alpha}) can be re-written as
		\begin{eqnarray*}
		L_2^{(\alpha)} (\theta) & = & \frac{1}{n}\sum\limits_{j=1}^n\bigg[\frac{\alpha\big\{Q(X_j)^{\alpha-1} + (1-\alpha)[Z(\theta)+\theta^Tf(X_j)]\big\} - 1}{\alpha-1}\bigg] - \sum\limits_{x\in\mathcal{X}}P_\theta(x)^\alpha\\
		& = & \frac{\alpha}{\alpha-1} \overline{Q^{\alpha-1}} - \Big[\alpha Z(\theta) + \frac{1}{\alpha-1}  + \alpha \theta^T\bar{f} + \sum\limits_{x\in\mathcal{X}}P_\theta(x)^\alpha\Big]\\ 
		& = & h(X_1,\dots, X_n) + g(\theta, T_2(X_1,\dots,X_n)),
		\end{eqnarray*}
		where
		\begin{flushleft}
			$h(X_1,\dots, X_n) := \tfrac{\alpha}{\alpha-1} \overline{Q^{\alpha-1}}$,\\
			$T_2(X_1,\dots,X_n) := \bar{f}$,\\
			and
			$g(\theta, T_2(X_1,\dots,X_n)) := -\Big[\alpha Z(\theta) + \frac{1}{\alpha-1}  + \alpha \theta^T\bar{f} + \sum\limits_{x\in\mathcal{X}}P_\theta(x)^\alpha\Big]$.
		\end{flushleft}	
		Hence $T_2(X_1,\dots,X_n) = \bar{f}$ is a sufficient statistics for $\theta$.
		\vspace*{0.2cm}
		
		\item[(c)] For $P_\theta\in \mathbb{M}^{(\alpha)}$, using (\ref{1eqn:alpha_power_law_formula}), the likelihood function in (\ref{1eqn:likelihood_function_for_I_alpha}) can be re-written as
		\begin{eqnarray*}
		\lefteqn{L_3^{(\alpha)}(\theta)}\\
		& & = \frac{\alpha}{\alpha-1}\log\Big[\frac{1}{n}\sum\limits_{j=1}^n 
		Z(\theta)^{1-\alpha} \big\{ Q(X_j)^{\alpha-1} + (1-\alpha) \theta^T f(X_j)\big\}
		\Big]-\log\sum\limits_{x\in\mathcal{X}}P_\theta(x)^\alpha\\
		& & = -\alpha \log Z(\theta) +\frac{\alpha}{\alpha-1} \log \big[ \overline{Q^{\alpha-1}}
		+ (1-\alpha) \theta^T\bar{f}\big]
		- \log\sum\limits_{x\in\mathcal{X}}P_\theta(x)^\alpha\\
		& & =  - \alpha \log Z(\theta) + \frac{\alpha}{\alpha-1} \log \overline{Q^{\alpha-1}} + 
		\frac{\alpha}{\alpha-1} \log \big[ 1+ (1-\alpha) \big\{{\theta^T\bar{f}}/
		{\overline{Q^{\alpha-1}}}\big\}\big]\\
		&& \hspace*{8cm} - \log\sum\limits_{x\in\mathcal{X}}P_\theta(x)^\alpha\\
		& & =  h(X_1,\dots, X_n) + g\big(\theta, T_3(X_1,\dots,X_n)\big),
		\end{eqnarray*}
		where
			$T_3(X_1,\dots,X_n) := {\bar{f}}/{\overline{Q^{\alpha-1}}}$, $h(X_1,\dots, X_n) := \frac{\alpha}{\alpha-1} \log \overline{Q^{\alpha-1}}$, and
			\begin{eqnarray*}
			g(\theta, T_3(X_1,\dots,X_n)) := - \alpha \log Z(\theta) + \frac{\alpha}{\alpha-1} \log \big[ 1+ (1-\alpha) \big\{{\theta^T\bar{f}}/{\overline{Q^{\alpha-1}}}\big\}\big]\\	- \log\sum\limits_{x\in\mathcal{X}}P_\theta(x)^\alpha.
			\end{eqnarray*}
		Hence $T_3(X_1,\dots,X_n) = {\bar{f}}/{\overline{Q^{\alpha-1}}}$ is a sufficient statistics for $\theta$.
		\vspace*{0.2cm}
		
		\item[(d)] For $P_\theta\in \mathscr{E}_\alpha$, using (\ref{1eqn:E_alpha_family_formula}), the likelihood function in (\ref{1eqn:likelihood_function_for_D_alpha}) can be re-written as
		\begin{eqnarray*}
		L_1^{(\alpha)}(\theta)	& = & \dfrac{1}{1- \alpha}\log\Big[\sum\limits_{x\in\mathcal{X}}\widehat{P}(x)^\alpha Z(\theta)^ {\alpha - 1} \big\lbrace Q(x)^{1-\alpha} +(1-\alpha)\theta^T f(x)\big\rbrace\Big]\\
		& = & -\log Z(\theta) + \dfrac{1}{1- \alpha}\log\Big[1 + (1-\alpha) {\theta^T{\overline{f}}^{(\alpha)}}/{{\overline{Q^{1-\alpha}}}^{(\alpha)}}\Big]\\
		& & \hspace*{3cm} + \dfrac{1}{1- \alpha}\log\sum\limits_{x\in\mathcal{X}}\widehat{P}(x)^\alpha Q(x)^{1-\alpha}\\
		& = & h(X_1,\dots,X_n) + g(\theta, T_4(X_1,\dots, X_n)),
		\end{eqnarray*}
		where
		
			$T_4(X_1,\dots,X_n) := {{\overline{f}}^{(\alpha)}}/{{\overline{Q^{1 - \alpha}}}^{(\alpha)}}$ with ${\overline{f}}^{(\alpha)}$ 
			and ${\overline{Q^{1 - \alpha}}}^{(\alpha)}$ being,  respectively, averages of $f(\cdot)$ and $Q(\cdot)^{ 1 - \alpha}$ with respect to 
			the measure $\widehat{P}^{(\alpha)}$,
			\vspace*{0.2cm}
			
			$g(\theta, T_4(X_1,\dots,X_n)) := -\log Z(\theta) + \dfrac{1}{1- \alpha}\log\Big[1 + (1-\alpha)
			\left\{{\theta^T{\overline{f}}^{(\alpha)}} / {{\overline{Q^{1 - \alpha}}}^{(\alpha)}}\right\}\Big]$, and
			\vspace*{0.2cm}
			
			$h(X_1,\dots,X_n) := \tfrac{1}{1- \alpha}\log {\sum\limits_{x\in\mathcal{X}} \widehat{P} (x)^\alpha
				Q(x)^{1 - \alpha}}$.
	
		Hence $T_4(X_1,\dots,X_n) = {{\overline{f}}^{(\alpha)}}/{{\overline{Q^{1 - \alpha}}}^{(\alpha)}}$ is a 
		sufficient statistic for $\theta$.
	\end{itemize}
\end{proof}
Thus the sufficient statistics $\bar{f}$, ${\bar{f}}/{\overline{Q^{\alpha-1}}}$, and ${{\overline{f}}^{(\alpha)}}/{{\overline{Q^{1 - \alpha}}}^{(\alpha)}}$ are precisely the statistics of the sample that influence the projection equations (\ref{1eqn:projection_equation_of_KL_div}), (\ref{1eqn:projection_equation_of_I_alpha_div}), and (\ref{1eqn:projection_equation_of_D_alpha_div1}) respectively.
\section{Summary}
\label{1sec:summary}

In this paper we studied four estimation problems on four statistical models distinctively associated with them. Each of these estimation problems can be regarded as a reverse projection problem of one of the divergences $I$, $B_{\alpha}$, $\mathscr{I}_{\alpha}$ and $D_{\alpha}$, on one of these statistical models. Projection theorems tell us that, the reverse projection, if exists, is the unique solution of an equation, called projection equation. While the projection equations are simple and easy to interpret in the case of $I$ and $B_{\alpha}$-divergences, it is not so in the case of $\mathscr{I}_{\alpha}$ and $D_{\alpha}$-divergences. Our objective in this paper was to understand these projection theorems (that is, projection equations) from the usual statistical method of solving by estimating equations. In the case of $I$ and $B_\alpha$-divergences, we saw that the projection equations and the estimating equations are the same and both depended only on the statistics of $f$ which is one of the defining entity of the underlying statistical model. However, in the case of $\mathscr{I}_\alpha$ and $D_\alpha$-divergences, the estimating equations and the projection equations are different and both depended not only on the statistics of $f$ but also on $Q$ which is other defining entity of the model. Nevertheless, we showed that, in the case of $\mathscr{I}_\alpha$ and $D_\alpha$ also, the solutions of the estimation problems obtained both from the projection equations and from the estimating equations are the same. We then tried to understand the projection theorems via sufficient statistics in the sense of generalized likelihood functions and observed that the estimating equations (or the projection equations) depended on the given sample only through these sufficient statistics. 

\appendix

\section{Projection Theorem for Density Power Divergence}
\label{1sec:the_projection_geometry_of_density_power_divergence}
The Projection theorem and the Pythagorean property of the more general class of Bregman divergences were established by Csisz\'ar and Mat\'u\v{s} \cite{CsiszarM12J} in generality using tools from convex analysis. The density power divergences $B_\alpha$ is a subclass of the Bregman divergences. However, it is hard to extract the results for the $B_\alpha$-divergence from \cite{CsiszarM12J}.  Hence we derive those results for the $B_\alpha$-divergence using some elementary tools as in \cite{CsiszarS04B} for the KL-divergence. We must point out that the geometry of $B_\alpha$-divergence is quite a natural extension of that of KL-divergence. Here we assume that $\mathcal{P} := \mathcal{P}(\mathcal{X})$ is the space of all probability measures on $\mathcal{X}$. Let $B_\alpha$ be as defined in (\ref{1defn:B_alpha_divergence})\footnote{We assume the usual convention that $B_\alpha(P,Q) = \infty$ when $P\not\ll Q$ and $\alpha <1$.}. Let us also recall the definitions of reverse and forward projections given in (\ref{1eqn:reverse_projection}) and (\ref{1eqn:forward_projection}). 

\begin{definition}
	For $P\in \mathcal{P}$, the support of $P$ is defined as $\text{Supp}(P) = \{ x\in \mathcal{X} : P(x)>0 \}$. For
	$\mathbb{C}\subset \mathcal{P}$, we will denote $\text{Supp}(\mathbb{C})$ for the union of support of members of $\mathbb{C}$. 
\end{definition}

First we show that, while the Pythagorean inequality is always satisfied when the forward projection is on a closed convex set, equality is satisfied when the forward projection is on a linear family and $\alpha <1$. Here, in the sequel, we assume that $\text{Supp}(Q) = \mathcal{X}$.

\begin{theorem}
	\label{1thm:pythagorean_theorem}
	Let $\alpha >0, \alpha\neq 1$. Let $P^*$ be the forward $B_\alpha$-projection of $Q$ on a closed and convex
	set $\mathbb{C}$. Then
	\begin{equation}
	\label{1eqn:Pythagorean_inequality}
	B_\alpha (P,Q) \geq B_\alpha (P,P^*)+B_\alpha(P^*,Q) \quad \forall P\in \mathbb{C}.
	\end{equation}
	Further, if $\alpha <1$, $\text{Supp}(\mathbb{C}) = \text{Supp}(P^*)$.
\end{theorem}

\begin{proof}
	Let $P\in\mathbb{C}$ and define
	\begin{center}
		$P_t(\cdot) = (1-t) P^*(\cdot) + t P(\cdot)$, for $t\in [0,1]$.
	\end{center}
	Since $\mathbb{C}$ is convex, $P_t\in \mathbb{C}$. By mean-value theorem, we have for each $t\in (0,1)$, 
	\begin{eqnarray}
	\label{1eqn:derivative}
	0 & \leq & \frac{1}{t}\big[B_\alpha (P_t, Q) -B_\alpha (P^*, Q)\big]\nonumber\\
	& = & \frac{1}{t}\big[B_\alpha (P_t, Q) -B_\alpha (P_0, Q)\big]\nonumber\\
	& = & \frac{d}{dt}B_\alpha (P_t, Q)\big |_{t=\tilde{t}}, \quad \text{for some } \tilde{t}\in (0,t).
	\end{eqnarray}
	Using definition of $B_\alpha$, we have
	\begin{eqnarray*}
		\dfrac{d}{dt} B_\alpha (P_t, Q)=\dfrac{\alpha}{\alpha-1}\sum\limits_{x\in \mathcal{X}} \big[P(x)-P^*(x)\big]\big[P_t(x)^{\alpha-1}-Q(x)^{\alpha-1}\big].
	\end{eqnarray*}
	Therefore, using (\ref{1eqn:derivative}),
	\begin{equation}
	\label{1eqn:derivative_calculation}
	\dfrac {\alpha}{\alpha-1}\sum\limits_{x\in \mathcal{X}} \big[P(x)-P^*(x)\big]\big[P_{\tilde{t}} (x)^{\alpha-1}-Q(x)^{\alpha-1}\big] \geq 0.
	\end{equation}
	Hence, as $t\downarrow 0$, we have
	\begin{equation}
	\label{1eqn:pythagorean_equivalent}
	\dfrac {\alpha}{\alpha-1}\sum\limits_{x\in \mathcal{X}} \big[P(x)-P^*(x)\big]\big[P^*(x)^{\alpha-1}-Q(x)^{\alpha-1}\big] \geq 0,
	\end{equation}
	which is equivalent to (\ref{1eqn:Pythagorean_inequality}).
	
	If $\text{Supp}(P^*)\neq \text{Supp}(\mathbb{C})$, that is, if $P^*(x) = 0$ for some $x\in \mathcal{X}$ and some $P\in \mathbb{C}$ such that $x\in \text{Supp}(P)$, and if $\alpha <1$, then the left-hand side of (\ref{1eqn:derivative_calculation}) goes to $-\infty$ as $t\downarrow 0$, which contradicts (\ref{1eqn:derivative_calculation}). This proves the claim.
\end{proof}

If $\alpha>1$, in general, $\text{Supp}(P^*)\neq \text{Supp}(\mathbb{C})$. \cite[Example 2]{KumarS15J2} serves as a counterexample here as well.

We will now show that equality holds in (\ref{1eqn:Pythagorean_inequality}) for $\alpha <1$ when $\mathbb{C}$ is a linear family, which we shall define now.
\begin{definition}
	\label{1defn:linear_family}
	The {\em linear family}, determined by $k$ real valued functions $f_i, i=1,\dots, k$ on $\mathcal{X}$ and $k$ real numbers $a_i, i=1,\dots,k$, is defined as
	\begin{equation}
	\label{1eqn:linear_family}
	\mathbb{L} := \Big\{ P\in \mathcal{P} : \small\sum\limits_{x\in \mathcal{X}} P(x)f_i(x) = a_i,\quad i=1,\dots, k \Big\}.
	\end{equation} 
\end{definition} 
\begin{theorem}
	Let $P^*$ be the forward $B_\alpha$-projection of $Q$ on $\mathbb{L}$. The following hold.
	\begin{enumerate}
		\item[(a)] If $\alpha<1$, then the Pythagorean equality holds, that is,
		\begin{equation}
		\label{1eqn:pythagorean_equality}
		B_\alpha (P,Q) = B_\alpha (P,P^*)+B_\alpha(P^*,Q) \quad\forall P\in \mathbb{L}.
		\end{equation}
		
		\item[(b)] If $\alpha>1$ and if $\text{Supp}(P^*) = \text{Supp}(\mathbb{L})$, then the Pythagorean equality (\ref{1eqn:pythagorean_equality}) holds.
	\end{enumerate}
\end{theorem}
\begin{proof}
	(a) Let $P_t$ be as in Theorem \ref{1thm:pythagorean_theorem}. Since $\text{Supp}(P^*) = \text{Supp}(\mathbb{L})$, there exists $t' < 0$ such that $P_t = (1 -t)P^* + tP\in \mathbb{L}$ for $t\in (t',0)$. Hence, for every $t\in (t',0)$, there exists $\tilde{t}\in (t,0)$ such that
	\begin{equation*}
	\frac {\alpha}{\alpha-1}\small\sum\limits_{x\in \mathcal{X}} \big[P(x)-P^*(x)\big]\big[P_{\tilde{t}} (x)^{\alpha-1}-Q(x)^{\alpha-1}\big] \le 0.
	\end{equation*}
	Proceeding as in Theorem \ref{1thm:pythagorean_theorem}, we get (\ref{1eqn:pythagorean_equivalent}) with a reversed inequality. Thus we have equality in (\ref{1eqn:pythagorean_equivalent}). Hence (\ref{1eqn:pythagorean_equality}) holds.
	\vspace*{0.2cm}
	
	(b) Similar to (a).
\end{proof}

\begin{remark}
	When $\alpha >1$, equality in (\ref{1eqn:pythagorean_equality}) does not hold in general. \cite[Example 2]{KumarS15J2} serves as a counterexample here as well.
\end{remark}

We will now find an explicit expression of the forward $B_\alpha$-projection in both the cases $\alpha<1$ and $\alpha >1$ separately.
\begin{theorem}
	\label{1thm:forward_projection}
	Let $Q\in \mathcal{P}$ and consider a linear family $\mathbb{L}$ of probability measures as in (\ref{1eqn:linear_family}).
	\begin{enumerate}
		\item[(a)] If $\alpha<1$, the forward $B_\alpha$-projection $P^*$ of $Q$ on $\mathbb{L}$ satisfies
		\begin{eqnarray}
		\label{1eqn:forward_projection_alpha<1}
		P^*(x) = \big[Q(x)^{\alpha-1}+(1-\alpha)\big\lbrace Z+\sum\limits_{i=1}^k \theta_if_i(x)\big\rbrace\big]^{\frac{1}{\alpha-1}} \quad\forall x\in \text{Supp}(\mathbb{L}),
		\end{eqnarray}
		where $\theta_1,\dots,\theta_k$ are scalars and $Z$ is a constant.
		\item[(b)] If $\alpha>1$, the forward $B_\alpha$-projection $P^*$ of $Q$ on $\mathbb{L}$ satisfies
		\begin{eqnarray}
		\label{1eqn:forward_projection_alpha>1}
		P^*(x) = \big[Q(x)^{\alpha-1}+(1-\alpha)\big\lbrace Z+\sum\limits_{i=1}^k \theta_if_i(x)
		\big\rbrace\big]_+^{\frac{1}{\alpha-1}} \quad\forall x\in \mathcal{X},
		\end{eqnarray}
		where $\theta_1,\dots,\theta_k$ are scalars, $Z$ is a constant, and for any real number $r$, $[r]_+:=\max \lbrace r,0\rbrace$.
	\end{enumerate}
	
\end{theorem}

\begin{proof}
	(a) The linear family in (\ref{1eqn:linear_family}) can be re-written as
	\begin{equation}
	\label{1eqn:linear_family_equivalent}
	\mathbb{L}:= \Big\{P\in\mathcal{P} : \sum\limits_{x\in \text{Supp}(\mathbb{L})} P(x)f_i(x) = a_i, \quad i=1,\dots,k\Big\}.
	\end{equation}
	Let $\mathcal{F}$ be the subspace of $\mathbb{R}^{|\text{Supp}(\mathbb{L})|}$ spanned by the $k$ vectors $f_1(\cdot)-a_1, \dots, f_k(\cdot)-a_k$. Then every $P\in \mathbb{L}$ can be thought of a $|\text{Supp}(\mathbb{L})|$-dimensional vector in $\mathcal{F}^{\perp}$. Hence $\mathcal{F}^{\perp}$ is a subspace of $\mathbb{R}^{|\text{Supp}(\mathbb{L})|}$ that contains a vector whose components are strictly positive as $P^*\in\mathbb{L}$. It follows that $\mathcal{F}^{\perp}$ is spanned by its	probability vectors. From (\ref{1eqn:pythagorean_equivalent}), we see that (\ref{1eqn:pythagorean_equality}) is equivalent to
	\begin{equation}
	\label{1eqn:pythagorean_equality_equivalent}
	\sum\limits_{x\in \mathcal{X}} \big[P(x)-P^*(x)\big]\big[P^*(x)^{\alpha-1}-Q(x)^{\alpha-1}\big] = 0\quad\forall P\in\mathbb{L}.
	\end{equation}
	This implies that the vector
	\begin{equation*}
	P^*(\cdot)^{\alpha-1} - Q(\cdot)^{\alpha-1} - \sum\limits_x P^*(x)\big[P^*(x)^{\alpha-1}-Q(x)^{\alpha-1}
	\big] \in (\mathcal{F}^\perp)^\perp=\mathcal{F}.
	\end{equation*}	
	Hence
	\begin{eqnarray*}
		P^*(x)^{\alpha-1} - Q(x)^{\alpha-1} - \sum\limits_x P^*(x)\big[P^*(x)^{\alpha-1} - Q(x)^{\alpha-1}\big]
		= \sum\limits_{i=1}^k c_i\big[f_i(x)-a_i\big]\\
		\forall x\in \text{Supp}(\mathbb{L}),
	\end{eqnarray*}
	for some scalars $c_1,\dots,c_k$. This implies (\ref{1eqn:forward_projection_alpha<1}) for appropriate choices of $Z$ and $\theta_1,\dots,\theta_k$.
	\vspace*{0.2cm}
	
	(b) The proof of this case is similar to the proof for relative $\alpha$-entropy \cite[Th. 14 (b)]{KumarS15J2}. The optimization problem involved in the forward $B_\alpha$-projection is
	\begin{align}
	\min_P \, & B_{\alpha}(P,Q)\label{1eqn:min}\\
	\mbox{subject to } &\sum\limits_x P(x)f_i(x) = a_i, \quad i=1, \dots, k,\label{1eqn:linear_constraints}\\
	&\sum\limits_x P(x)       = 1, \label{1eqn:probability_constraint}\\
	&P(x)                    \ge 0 \quad\forall x \in \mathcal{X}. \label{1eqn:positivity_constraints}
	\end{align}
	Hence, by \cite[Prop.~3.3.7]{Bertsekas03B}, there exists Lagrange multipliers $\lambda_1,\dots,\lambda_k$, 
	$\nu$, and $(\mu(x),x\in \mathcal{X})$ respectively associated with the above constraints such that, for $x\in \mathcal{X}$,
	\begin{eqnarray}
	\label{1eqn:lagrange}
	\dfrac{\partial}{\partial P(x)}B_\alpha(P,Q)\Big|_{P=P^*} & = & \sum\limits_{i=1}^k \lambda_i [f_i(x) - a_i] + \mu(x) - \nu,\\
	\label{1eqn:feasibility_condition}
	\mu(x) & \ge & 0,\\
	\label{1eqn:slackness_condition}
	\mu(x)P^*(x) & = & 0.
	\end{eqnarray}
	Since
	\begin{equation}
	\label{1eqn:partial_derivative}
	\dfrac{\partial}{\partial P(x)}B_\alpha(P,Q)=\dfrac{\alpha}{\alpha-1}
	\big[P(x)^{\alpha-1}-Q(x)^{\alpha-1}\big],
	\end{equation}
	(\ref{1eqn:lagrange}) can be re-written as
	\begin{eqnarray}
	\label{1eqn:lagrange1}
	\dfrac{\alpha}{\alpha-1}\big[P^*(x)^{\alpha-1} - Q(x)^{\alpha-1}\big] = \sum\limits_{i=1}^k\lambda_i\big[f_i(x)-a_i\big] + \mu(x) - \nu\quad \text{ for } x\in \mathcal{X}.
	\end{eqnarray}
	Multiplying both sides by $P^*(x)$ and summing over all $x\in \mathcal{X}$, we get
	\begin{equation*}
	\nu = \dfrac{\alpha}{\alpha-1}\sum\limits_{x\in \mathcal{X}} P^*(x)\big[ Q(x)^{\alpha-1}-P^*(x)^{\alpha-1}\big].
	\end{equation*}
	For $x\in \text{Supp}(P^*)$, from (\ref{1eqn:slackness_condition}), we must have $\mu(x)=0$. Then, from (\ref{1eqn:lagrange1}), we have
	\begin{eqnarray}
	\label{1eqn:projection_alpha>1_equivalent1}
	P^*(x)^{\alpha-1} = Q(x)^{\alpha-1} + \dfrac{\alpha-1}{\alpha} \sum\limits_{i=1}^k\lambda_i
	\big[f_i(x)-a_i\big] - \dfrac{\alpha-1}{\alpha}\nu.
	\end{eqnarray}
	If $P^*(x)=0$, from (\ref{1eqn:lagrange1}), we get
	\begin{equation}
	\label{1eqn:projection_alpha>1_equivalent2}
	Q(x)^{\alpha-1} + \dfrac{\alpha-1}{\alpha} \sum\limits_{i=1}^k\lambda_i\big[f_i(x)-a_i\big] - \dfrac{\alpha-1}{\alpha}\nu = -
	\dfrac{\alpha-1}{\alpha}\mu(x)\leq 0.
	\end{equation}
	Combining (\ref{1eqn:projection_alpha>1_equivalent1}) and (\ref{1eqn:projection_alpha>1_equivalent2}) we get (\ref{1eqn:forward_projection_alpha>1}).
\end{proof}
\begin{remark}
	The expression for $P^*$ in (\ref{1eqn:forward_projection_alpha<1}) (or in (\ref{1eqn:forward_projection_alpha>1})) is given with the multiplicative factor $1-\alpha$ in the square brackets so that the formula in (\ref{1eqn:forward_projection_alpha<1}) can be re-written, with the help of {\em $\alpha$-exponential} and {\em $\alpha$-logarithmic} functions \cite[Defn.~7]{KumarS15J2}, as $P^*(x)^{-1} = e_\alpha\big[\ln_\alpha(Q(x)^{-1}) + Z +  \small\sum\limits_{i=1}^k\theta_i f_i(x)\big]$, analogous to the exponential probability measure (see (\ref{1eqn:exponential_family_formula})).
\end{remark}
Theorem \ref{1thm:forward_projection} suggests us to define a parametric family of probability measures that extends the usual exponential family. We first formally define this family and then show an orthogonality relationship between this family and the linear family. As a consequence we will also show that the reverse $B_{\alpha}$-projection on this generalized exponential family is same as a forward projection on a linear family.
\begin{definition}
	\label{1defn:non_normalized_alpha_powerlaw_family}
	The {\em non-normalized $\alpha$-power-law family}\footnote{This term is coined due to the fact that this family differs from the $\alpha$-power-law family only in the normalizing constant.}, $\mathbb{B}^{(\alpha)} := \mathbb{B}^{(\alpha)}(Q, f, \Theta)$, characterized by a $Q\in\mathcal{P}$, $k$ real valued functions $f_i, i=1,\dots, k$ on $\mathcal{X}$, and a parameter space $\Theta\subset\mathbb{R}^k$, is given by $\mathbb{B}^{(\alpha)} = \{P_\theta : \theta\in\Theta\}\subset\mathcal{P}$, where
	\begin{eqnarray}
	\label{1eqn:B_alpha_family_formula}
	P_{\theta}(x) = \big[Q(x)^{\alpha-1} + (1-\alpha) \big\{Z(\theta) + \theta^Tf(x)\big\}\big]^{\frac{1}{\alpha-1}} \quad\text{for } x\in\mathcal{X}.
	\end{eqnarray}
\end{definition}

\begin{remark}
	\label{1rem:B_alpha_family}
	\begin{itemize}
		\item[(a)] Observe that $\mathbb{B}^{(\alpha)}$ is a special case of the family $\mathcal{F}_{[\beta h]}$ in \cite[Eq.~(28)]{CsiszarM12J} when $h(\cdot) = Q(\cdot)$ and $\beta(\cdot, t) = \frac{1}{\alpha -1}[t^{\alpha}-\alpha t +\alpha -t]$.
		\vspace*{0.2cm}
		
		\item[(b)] As in the case of $\mathcal{E}$, $\mathbb{M}^{(\alpha)}$, and $\mathscr{E}_\alpha$ families, the family $\mathbb{B}^{(\alpha)}$ depends on the reference measure $Q$ only in a loose manner in the sense that any other member of $\mathbb{B}^{(\alpha)}$ can play the role of $Q$. The change of reference measure only corresponds to a translation of the parameter space.
	\end{itemize}
\end{remark}

The following theorem and its corollary together establish an ``orthogonality" relationship between the non-normalized $\alpha$-power-law family and its associated linear family.
\begin{theorem}
	\label{1thm:orthogonality1}
	Let $\alpha<1$. Let $\mathbb{B}^{(\alpha)}$ be the non-normalized $\alpha$-power-law family as in Definition \ref{1defn:non_normalized_alpha_powerlaw_family} and $\mathbb{L}$ be the corresponding linear family determined by the same functions $f_i, i = 1,\dots, k$ and some constants $a_i,i=1,\dots,k$ as in (\ref{1eqn:linear_family}). If $P^*$ is the forward $B_\alpha$-projection of $Q$ on $\mathbb{L}$, then we have the following:
	\begin{enumerate}
		\item[(a)] $\mathbb{L}\cap \text{cl}(\mathbb{B}^{(\alpha)}) = \{ P^*\}$ and
		\begin{eqnarray}
		\label{1eqn:pythagorean_equality1}
		B_\alpha (P,Q) = B_\alpha (P,P^*) + B_\alpha(P^*,Q)\quad \forall P\in \mathbb{L}.
		\end{eqnarray}
		\item[(b)] Further, if $\text{Supp}(\mathbb{L}) = \mathcal{X}$, then $\mathbb{L}\cap\mathbb{B}^{(\alpha)}=\{ P^*\}$.
	\end{enumerate}
\end{theorem}
\begin{proof}
	By Theorem \ref{1thm:forward_projection}, the forward $B_\alpha$-projection $P^*$ of $Q$ on $\mathbb{L}$ is in $\mathbb{B}^{(\alpha)}$. This implies that $P^*\in \mathbb{L}\cap \mathbb{B}^{(\alpha)}$. In general, it would suffice if we prove the following:
	\begin{enumerate}
		\item[(i)] Every $\tilde{P}\in \mathbb{L}\cap \text{cl}(\mathbb{B}^{(\alpha)})$ satisfies (\ref{1eqn:pythagorean_equality}) with $\tilde{P}$ in place of $P^*$.
		\item[(ii)] $\mathbb{L}\cap \text{cl}(\mathbb{B}^{(\alpha)})$ is non-empty.
	\end{enumerate}
	
	We now proceed to prove both (i) and (ii).
	\vspace*{0.2cm}
	
	(i) Let $\tilde{P}\in \mathbb{L}\cap \text{cl}(\mathbb{B}^{(\alpha)})$. As $\tilde{P}\in\text{cl}(\mathbb{B}^{(\alpha)})$, this implies there exists a sequence $\{P_n\}\subset\mathbb{B}^{(\alpha)}$ such that $P_n\rightarrow \tilde{P}$ as $n\to\infty$. Since $P_n\in \mathbb{B}^{(\alpha)}$, we can write
	\begin{equation}
	\label{1eqn:P_n}
	P_n(x)^{\alpha-1} = Q(x)^{\alpha-1} + (1-\alpha)\big[Z_n + \sum\limits_{i=1}^k\theta_i^{(n)}f_i(x)\big]\quad\forall x\in \mathcal{X}
	\end{equation}
	for some constants $\theta^{(n)} = (\theta_1^{(n)},\dots,\theta_k^{(n)})\in \mathbb{R}^k$
	and $Z_n$. Now, for any $P\in \mathbb{L}$ we have, from the definition of linear family, $\sum\limits_{x\in \mathcal{X}}P(x)f_i(x) = a_i, i=1,\dots,k$. Since $\tilde{P}\in \mathbb{L}$, we also have
	$\sum\limits_{x\in \mathcal{X}}\tilde{P} (x)f_i(x) = a_i, i=1,\dots,k$. Multiplying both sides of (\ref{1eqn:P_n}) by $P(\cdot)$ and $\tilde{P}(\cdot)$ separately, we get
	\begin{equation*}
	\sum\limits_{x\in \mathcal{X}} P(x) P_n(x)^{\alpha-1} = \sum\limits_{x\in \mathcal{X}}P(x)Q(x)^{\alpha-1} + (1-\alpha)\Big[Z_n + \sum\limits_{i=1}^k\theta_i^{(n)}a_i\Big]
	\end{equation*}
	and
	\begin{equation*}
	\sum\limits_{x\in \mathcal{X}} \tilde{P}(x) P_n(x)^{\alpha-1} = \sum\limits_{x\in \mathcal{X}}\tilde{P} (x)Q(x)^{\alpha-1} + (1-\alpha)\Big[Z_n + \sum\limits_{i=1}^k\theta_i^{(n)}a_i\Big],
	\end{equation*}
	which imply
	\begin{equation*}
	\sum\limits_{x\in \mathcal{X}}\big[P(x) -\tilde{P}(x)\big]\big[P_n(x)^{\alpha-1} -Q(x)^{\alpha-1}
	\big] = 0.
	\end{equation*}
	As $n\to\infty$, the above becomes
	\begin{equation*}
	\sum\limits_{x\in \mathcal{X}}\big[P(x) -\tilde{P}(x)\big]\big[\tilde{P}(x)^{\alpha-1} - Q(x)^{\alpha-1}\big] = 0,
	\end{equation*}
	which is equivalent to (\ref{1eqn:pythagorean_equality}).
	\vspace*{0.2cm}
	
	(ii) Let $P_n^*$ be the forward $B_\alpha$-projection of $Q$ onto the linear family
	\begin{equation*}
	\mathbb{L}_n :=
	\Big\{P:\sum\limits_{x\in \mathcal{X}} P(x)f_i(x) = \Big(1-\frac{1}{n}\Big)a_i+\frac{1}{n}\sum\limits
	_{x\in \mathcal{X}} Q(x)f_i(x), \quad i = 1,\dots, k\Big\}.
	\end{equation*}
	By construction, $\left(1-\frac{1}{n}\right) P+\frac{1}{n}Q\in \mathbb{L}_n$ for any $P\in \mathbb{L}$. Hence, since $\text{Supp}(Q) = \mathcal{X}$, we have $\text{Supp}(\mathbb{L}_n) = \mathcal{X}$. Since $\mathbb{L}_n$ is also characterized by the same functions $f_i, i = 1,\dots, k$, we have $P^*_n\in \mathbb{B}^{(\alpha)}$ for every $n\in \mathbb{N}$. Hence limit of any convergent subsequence of $\{ P^*_n\}$ belongs to $\text{cl}(\mathbb{B}^{(\alpha)})\cap \mathbb{L}$. Thus $\text{cl}(\mathbb{B}^{(\alpha)})\cap \mathbb{L}$ is non-empty. This completes the proof.
\end{proof}
\begin{corollary}
	\label{1cor:orthogonality}
	Let $\mathbb{L}$ and $\mathbb{B}^{(\alpha)}$ be characterized by the same functions $f_i, i = 1,\dots, k$. Then $\mathbb{L}\cap \text{cl}(\mathbb{B}^{(\alpha)}) = \{P^*\}$ and
	\begin{eqnarray}
	\label{1eqn:orthogonality}
	B_\alpha (P,Q) = B_\alpha (P,P^*)+B_\alpha(P^*,Q)\quad\forall P\in \mathbb{L},\quad\forall Q\in \text{cl}(\mathbb{B}^{(\alpha)}).
	\end{eqnarray}
\end{corollary}
\begin{proof}
	By Theorem \ref{1thm:orthogonality1}, we have $\mathbb{L}\cap \text{cl}(\mathbb{B}^{(\alpha)}) = \{P^*\}$. In view of Remark \ref{1rem:B_alpha_family}(b), notice that every member of $\mathbb{B}^{(\alpha)}$ has the same projection on $\mathbb{L}$, namely $P^*$. Hence (\ref{1eqn:orthogonality}) holds for every $Q\in\mathbb{B}^{(\alpha)}$. Thus we only need to prove (\ref{1eqn:orthogonality}) for every $Q\in\text{cl}(\mathbb{B}^{(\alpha)})\setminus\mathbb{B}^{(\alpha)}$. Let $Q\in\text{cl}(\mathbb{B}^{(\alpha)})\setminus\mathbb{B}^{(\alpha)}$. There exists $\{Q_n\}\subset\mathbb{B}^{(\alpha)}$ such that $Q_n\rightarrow Q$. Hence, for any $P\in \mathbb{L}$,
	\begin{equation}
	B_\alpha (P,Q_n) = B_\alpha (P,P^*) + B_\alpha (P^*,Q_n)\quad \forall n\in\mathbb{N}.
	\end{equation}
	Taking limit as $n\to\infty$, we have (\ref{1eqn:orthogonality})\footnote{For a fixed $P$, $Q\mapsto B_\alpha (P, Q)$ is continuous as a function from $\mathcal{P}$ to $[0, \infty]$.}. This completes the proof.	
\end{proof}
The following theorem tells us that a reverse $B_\alpha$-projection on a non-normalized $\alpha$-power-law family can be turned into a forward $B_\alpha$-projection on the associated linear family. We shall refer this as the {\em projection theorem} for the $B_\alpha$-divergence.
\begin{theorem}
	\label{1thm:orthogonality2}
	Let $\alpha<1$. Let $\mathbb{B}^{(\alpha)}$ be as in Definition \ref{1defn:non_normalized_alpha_powerlaw_family} and let $X_1,\dots, X_n$ be an i.i.d. sample drawn according to a particular member of $\mathbb{B}^{(\alpha)}$. Let $\widehat{P}$ be the empirical probability measure of $X_1,\dots, X_n$ and let
	\begin{equation}
	\label{1eqn:empirical_linear_family}
	\widehat{\mathbb{L}} := \Big\{P\in\mathcal{P} : \sum\limits_{x\in \mathcal{X}} P(x)f_i(x) = \bar{f_i}, \quad i=1,\dots,k\Big\},
	\end{equation}
	where $\bar{f_i} = \frac{1}{n}\sum_{j=1}^n f_i(X_j), i=1,\dots,k$. Let $P^*$ be the forward $B_\alpha$-projection of $Q$ on $\widehat{\mathbb{L}}$. Then the following hold.
	\begin{enumerate}
		\item[(i)] If $P^*\in \mathbb{B}^{(\alpha)}$, then $P^*$ is the	reverse $B_\alpha$-projection of $\widehat{P}$ on $\mathbb{B}^{(\alpha)}$.
		\item[(ii)] If $P^*\notin \mathbb{B}^{(\alpha)}$, then $\widehat{P}$ does not have a reverse $B_\alpha$-projection on $\mathbb{B}^{(\alpha)}$. However, $P^*$ is the reverse $B_\alpha$-projection of $\widehat{P}$ on \text{cl}$(\mathbb{B}^{(\alpha)})$.
	\end{enumerate}
\end{theorem}
\begin{proof}
	Let us first observe that $\widehat{\mathbb{L}}$ is so constructed that $\widehat{P}\in\widehat{\mathbb{L}}$. Since the families $\widehat{\mathbb{L}}$ and $\mathbb{B}^{(\alpha)}$ are defined by the same functions $f_i$, $i=1,\dots,k$, by Corollary \ref{1cor:orthogonality}, we have $\widehat{\mathbb{L}}\cap \text{cl}(\mathbb{B}^{(\alpha)}) = \{P^*\}$ and 
	\begin{equation}
	\label{1eqn:pythagorean_equality2}
	B_\alpha(\widehat{P},Q) = B_\alpha (\widehat{P},P^*) + B_\alpha(P^*,Q) \quad\forall Q\in \text{cl}(\mathbb{B}^{(\alpha)}).
	\end{equation}
	Hence it is clear that the minimizer of $B_\alpha(\widehat{P}, Q)$ over $Q\in\text{cl}(\mathbb{B}^{(\alpha)})$ is same as the minimizer of $B_\alpha (P^* , Q)$ over $Q\in\text{cl}(\mathbb{B}^{(\alpha)})$ (Notice that this statement is also true with $\text{cl}(\mathbb{B}^{(\alpha)})$ replaced by $\mathbb{B}^{(\alpha)}$). But $B_\alpha (P^* , Q)$ over $Q\in\text{cl}(\mathbb{B}^{(\alpha)})$ is uniquely minimized by $Q = P^*$. Hence, if $P^*\notin \mathbb{B}^{(\alpha)}$, since minimum of $B_\alpha(\widehat{P}, Q)$ over $Q\in\text{cl}(\mathbb{B}^{(\alpha)})$ is same as the minimum of $B_\alpha(\widehat{P}, Q)$ over $Q\in\mathbb{B}^{(\alpha)}$, the later is not attained on $\mathbb{B}^{(\alpha)}$.	
\end{proof}

\begin{remark}
	Theorems \ref{1thm:orthogonality1}, \ref{1thm:orthogonality2}, and Corollary \ref{1cor:orthogonality} continue to hold for $\alpha >1$ as well if attention is restricted to probability measures with strictly positive components and the existence of $P^*$ is guaranteed.
\end{remark}

\section{The Map $\theta\mapsto\Phi(\theta)$ Defined in (\ref{1eqn:expression_of_phi_i_1}) is one-one:}
\label{1sec:phi_is_oneone}
Let us first recall that $\Phi = (\phi_1,\dots,\phi_k)$, where
\begin{equation*}
\phi_i(\theta) = \frac{\mathbb{E}_\theta [f_i(X)]\overline{P_\theta^{\alpha-1}}}{\|P_\theta\|^\alpha}, \quad\text{ for } i = 1,\dots,k.
\end{equation*}
We will show that the determinant of the Jacobian matrix of $\Phi$ is non-zero. Then the claim follows from \cite[Th. 13.6]{Apostol02B}. Let $\partial_j$ denote $\partial/\partial\theta_j$. Then the $(i,j)$th entry of the Jacobian matrix of $\Phi$ is given by
\begin{eqnarray}
\label{1eqn:jacobian_matrix}
\partial_j\phi_i & = & \overline{{P}_\theta^{\alpha-1}}\bigg\{\dfrac{\|P_\theta\|^\alpha\partial_j(\mathbb{E}_\theta[f_i(X)]) - \mathbb{E}_\theta[f_i(X)]\partial_j (\|P_\theta\|^\alpha)}{\|P_\theta\|^{2\alpha}}\bigg\}+\dfrac{\mathbb{E}_\theta[f_i(X)]}{\|P_\theta\|^\alpha} \partial_j(\overline{{P}_\theta^{\alpha-1}})\nonumber\\
& = & \dfrac{\overline{{P_\theta}^{\alpha-1}}}{\|P_\theta\|^\alpha} \partial_j(\mathbb{E}_\theta[f_i(X)]) - \overline{{P}_\theta^{\alpha-1}}\dfrac{\mathbb{E}_\theta[f_i(X)]\partial_j (\|P_\theta\|^\alpha)}{\|P_\theta\|^{2\alpha}} + \dfrac{\mathbb{E}_\theta[f_i(X)]}{\|P_\theta\|^\alpha} \partial_j(\overline{{P}_\theta^{\alpha-1}})\nonumber\\
& = & \dfrac{\overline{{P_\theta}^{\alpha-1}}}{\|P_\theta\|^\alpha} \partial_j(\mathbb{E}_\theta [f_i(X)]) + \dfrac{\mathbb{E}_\theta[f_i(X)]}{\|P_\theta\|^\alpha}\left[-\overline{{P}_\theta^{\alpha-1}}\partial_j\big(\log
\|P_\theta\|^\alpha\big) + \partial_j(\overline{{P}_\theta^{\alpha-1}})\right].\nonumber\\
\end{eqnarray}
We now calculate each of the terms in (\ref{1eqn:jacobian_matrix}).
\begin{eqnarray*}
	\partial_j\big(\mathbb{E}_\theta[f_i(X)]\big) & = &\sum\limits_{x\in\mathcal{X}}\partial_j [P_\theta(x)]f_i(x)\\
	& = & \sum\limits_{x\in\mathcal{X}}\partial_j\Big(Z(\theta)^{-1}\big[Q(x)^{\alpha-1} + (1-\alpha)\theta^T f(x)\big]^{\frac{1}{\alpha-1}}\Big)f_i(x)\\
	& = & -\sum\limits_{x\in\mathcal{X}}Z(\theta)^{-1}P_\theta(x)f_i(x)\partial_j [Z(\theta)]\nonumber\\
	& & - Z(\theta)^{-1}\sum\limits_{x\in\mathcal{X}}\big[Q(x)^{\alpha-1} + (1-\alpha)\theta^T f(x)\big]^{\frac{1}{\alpha-1}-1}f_i(x) f_j(x)\\
	& = & -Z(\theta)^{-1}\mathbb{E}_\theta[f_i(X)] \partial_j [Z(\theta)]\nonumber\\
	& & - Z(\theta)^{-1}\sum\limits_{x\in\mathcal{X}}\big[Z(\theta)^{\alpha-1}P_\theta(x)^{\alpha-1}\big]^{\frac{1}{\alpha-1}-1}f_i(x)f_j(x)\\
	& = & -Z(\theta)^{-1}\mathbb{E}_\theta[f_i(X)] \partial_j[Z(\theta)]\nonumber\\
	& & - Z(\theta)^{1-\alpha}\sum\limits_{x\in\mathcal{X}}P_\theta(x)^\alpha[P_\theta(x)^{1-\alpha}f_i(x)][ P_\theta(x)^{1-\alpha}f_j(x)].
\end{eqnarray*}
\flushbottom
Hence
\begin{eqnarray}
\label{1eqn:1st_term}
\dfrac{\overline{{P_\theta}^{\alpha-1}}}{\|P_\theta\|^\alpha}\partial_j(\mathbb{E}_\theta[f_i(X)]) & = & -\overline{{P_\theta}^{\alpha-1}}Z(\theta)^{-1}\partial_j [Z(\theta)]\mathbb{E}_{\theta^{(\alpha)}}[P_\theta^{1-\alpha}(X)f_i(X)]\nonumber\\
& & - Z(\theta)^{1-\alpha}\overline{{P_\theta}^{\alpha-1}}\mathbb{E}_{\theta^{(\alpha)}} [P_\theta^{1-\alpha}(X)f_i(X)P_\theta^{1-\alpha}(X)f_j(X)],\nonumber\\
\end{eqnarray}
where we used the fact that
\begin{equation*}
\frac{\mathbb{E}_\theta[f_i(X)]}{\|P_\theta\|^\alpha} = \small\mathbb{E}_{\theta^{(\alpha)}}[P_\theta^{1-\alpha}(X)f_i(X)].
\end{equation*}
Similarly
\begin{eqnarray}
\label{1eqn:2nd_term}
\partial_j\big[\log\|P_\theta\|^\alpha\big] & = &\partial_j\Big[\log\small\sum\limits_{x\in\mathcal{X}}P_\theta(x)^\alpha\Big]\nonumber\\
& = & \partial_j\Big[\log\Big(Z(\theta)^{-\alpha}\sum\limits_{x\in\mathcal{X}}\big[Q(x)^{\alpha-1}+(1-\alpha)\theta^Tf(x)\big]^{\frac{\alpha}{\alpha-1}}\Big)\Big]\nonumber\\
& = & \partial_j\big(\log\big[Z(\theta)^{-\alpha}\big]\big) + \partial_j\Big(\log\sum\limits_{x\in\mathcal{X}}\big[Q(x)^{\alpha-1} + (1-\alpha)\theta^Tf(x)\big]^{\frac{\alpha}{\alpha-1}}\Big)\nonumber\\
& = & -\alpha Z(\theta)^{-1}\partial_j[Z(\theta)] - \dfrac{\alpha\sum\limits_{x\in\mathcal{X}} \big[Q(x)^{\alpha-1}+(1-\alpha)\theta^Tf(x)\big]^{\frac{1}{\alpha-1}}f_j(x)}{\sum\limits_{x\in\mathcal{X}}\big[Q(x)^{\alpha-1}+(1-\alpha)\theta^Tf(x)\big]
	^{\frac{\alpha}{\alpha-1}}}\nonumber\\
& = & -\alpha Z(\theta)^{-1}\partial_j[Z(\theta)]-\dfrac{\alpha}{Z(\theta)^\alpha\|P_\theta\|^\alpha}Z(\theta)\sum\limits_{x\in\mathcal{X}}P_\theta(x)f_j(x)\nonumber\\
& = & -\alpha Z(\theta)^{-1}\partial_j[Z(\theta)] - \alpha Z(\theta)^{1-\alpha}\dfrac{\mathbb{E}_\theta[f_j(X)]}{\|P_\theta\|^\alpha}.
\end{eqnarray}
Also
\begin{eqnarray}
\partial_j(\overline{{P}_\theta^{\alpha-1}}) & = & \partial_j\big(Z(\theta)^{1-\alpha}\big[\overline{Q^{\alpha-1}} + (1-\alpha)\theta^T\bar{f}\big]\big)\nonumber\\
& = & (1-\alpha)Z(\theta)^{-\alpha}\partial_j[Z(\theta)]\big[\overline{Q^{\alpha-1}}+(1-\alpha)\theta^T\bar{f}\big] + (1-\alpha)Z(\theta)^{1-\alpha}\bar{f}_j\nonumber\\
& = & (1-\alpha)Z(\theta)^{-1}\overline{{P_\theta}^{\alpha-1}} \partial_j[Z(\theta)] + (1-\alpha)Z(\theta)^{1-\alpha}\dfrac{\mathbb{E}_\theta[f_j(X)]\overline{{P}_\theta^{\alpha-1}}}{\|P_\theta\|^\alpha},\label{1eqn:3rd_term}\nonumber\\
\end{eqnarray}
where in (\ref{1eqn:3rd_term}) we used the fact that
\begin{equation*}
\bar{f}_j = \dfrac{\mathbb{E}_\theta[f_j(X)]\overline{{P}_\theta^{\alpha-1}}}{\|P_\theta\|^\alpha}
\end{equation*}
by (\ref{1eqn:estimating_equation_I_alpha_simplified}) and (\ref{1eqn:expression_of_phi_i_1}).

Using (\ref{1eqn:2nd_term}) and (\ref{1eqn:3rd_term}), the last two terms in (\ref{1eqn:jacobian_matrix}) together yield
\begin{eqnarray}
\label{1eqn:2nd_and_3rd_term_final}
\lefteqn{\dfrac{\mathbb{E}_\theta[f_i(X)]}{\|P_\theta\|^\alpha}\left[-\overline{{P}_\theta^{\alpha-1}}\partial_j\big(\log\|P_\theta\|^\alpha\big)+\partial_j\overline{{P}_\theta^{\alpha-1}}\right]}\nonumber\\
& & = \dfrac{\mathbb{E}_\theta[f_i(X)]}{\|P_\theta\|^\alpha}\Bigg[\alpha \overline{{P}_\theta^{\alpha-1}}Z(\theta)^{-1} \partial_j[Z(\theta)] +\alpha\overline{{P}_\theta^{\alpha-1}}Z(\theta)^{1-\alpha}\dfrac{\mathbb{E}_\theta[f_j(X)]}{\|P_\theta\|^\alpha}\nonumber\\
& & \hspace{1cm} + (1-\alpha)Z(\theta)^{-1}\overline{{P_\theta}^{\alpha-1}} \partial_j[Z(\theta)] + (1-\alpha)Z(\theta)^{1-\alpha}\dfrac{\mathbb{E}_\theta[f_j(X)]\overline{{P}_\theta^{\alpha-1}}}{\|P_\theta\|^\alpha}\Bigg]\nonumber\\
& & = Z(\theta)^{-1}\overline{{P_\theta}^{\alpha-1}}\partial_j[Z(\theta)]\dfrac{\mathbb{E}_\theta[f_i(X)]}{\|P_\theta\|^\alpha} + Z(\theta)^{1-\alpha}\overline{{P}_\theta^{\alpha-1}}\dfrac{\mathbb{E}_\theta[f_i(X)] }{\|P_\theta\|^\alpha}\dfrac{\mathbb{E}_\theta[f_j(X)]}{\|P_\theta\|^\alpha}\nonumber\\
& & = Z(\theta)^{-1}\overline{{P_\theta}^{\alpha-1}}\partial_j[Z(\theta)]\mathbb{E}_{\theta^{(\alpha)}}[P_\theta^{1-\alpha}(X)f_i(X)]\nonumber\\
& & \hspace{1cm} + Z(\theta)^{1-\alpha}\overline{{P}_\theta^{\alpha-1}}\mathbb{E}_{\theta^{(\alpha)}}[P_\theta^{1-\alpha
}(X)f_i(X)]\mathbb{E}_{\theta^{(\alpha)}}[P_\theta^{1-\alpha}(X)f_j(X)].
\end{eqnarray}
Substituting (\ref{1eqn:1st_term}) and (\ref{1eqn:2nd_and_3rd_term_final}) in (\ref{1eqn:jacobian_matrix}), we get
\begin{eqnarray}
\label{1eqn:jacobian_matrix_final}
\partial_j\phi_i & = & -\overline{{P_\theta}^{\alpha-1}}Z(\theta)^{-1}\partial_j [Z(\theta)]\mathbb{E}_{\theta^{(\alpha)}}[P_\theta^{1-\alpha}(X)f_i(X)]\nonumber\\
& & - Z(\theta)^{1-\alpha}\overline{{P_\theta}^{\alpha-1}}\mathbb{E}_{\theta^{(\alpha)}} [P_\theta^{1-\alpha}(X)f_i(X)P_\theta^{1-\alpha}(X)f_j(X)]\nonumber\\
& &  + Z(\theta)^{-1}\overline{{P_\theta}^{\alpha-1}}\partial_j[Z(\theta)]\mathbb{E}_{\theta^{(\alpha)}}[P_\theta^{1-\alpha}(X)f_i(X)]\nonumber\\
& &  + Z(\theta)^{1-\alpha}\overline{{P}_\theta^{\alpha-1}}\mathbb{E}_{\theta^{(\alpha)}}[P_\theta^{1-\alpha}(X)f_i(X)]\mathbb{E}_{\theta^{(\alpha)}}[P_\theta^{1-\alpha}(X)f_j(X)]\nonumber\\
& = & - Z(\theta)^{1-\alpha}\overline{{P}_\theta^{\alpha-1}}\mathbb{E}_{\theta^{(\alpha)}} [P_\theta^{1-\alpha}(X)f_i(X)P_\theta^{1-\alpha}(X)f_j(X)]\nonumber\\
& &  + Z(\theta)^{1-\alpha}\overline{{P}_\theta^{\alpha-1}}\mathbb{E}_{\theta^{(\alpha)}}[P_\theta^{1-\alpha}(X)f_i(X)]\mathbb{E}_{\theta^{(\alpha)}}[P_\theta^{1-\alpha}(X)f_j(X)]\nonumber\\
& = & - Z(\theta)^{1-\alpha}\overline{{P}_\theta^{\alpha-1}} \text{Cov}_{\theta^{(\alpha)}}[P_\theta^{1-\alpha}(X)f_i(X),P_\theta^{1-\alpha}(X)f_j(X)].
\end{eqnarray}
Let $A$ denote the above covariance matrix, that is, the $(i,j)$th entry of $A$ is
\begin{equation*}
\text{Cov}_{\theta^{(\alpha)}}[P_\theta^{1-\alpha}(X)f_i(X),P_\theta^{1-\alpha}(X)f_j(X)].
\end{equation*} 
Then
$A$ is symmetric and positive semi-definite. Let $B$ denote the Jacobian matrix of $\Phi$.
Then, from (\ref{1eqn:jacobian_matrix_final}), we have 
$B = - Z(\theta)^{1-\alpha}\overline{{P}_\theta^{\alpha-1}} A$. So, $B$ is also symmetric but negative
semi-definite. If we can show that $B$ is negative definite then then the determinant of $B$ is non-zero. Thus it suffices to show that $A$ is positive definite. Suppose that $A$ is not positive definite. Then there exists a non-zero
vector $c = (c_1, \dots, c_k)^T$ such that
\begin{eqnarray}
c^T A c = 0.
\end{eqnarray}
This implies 
\begin{eqnarray}
\text{Var}_{\theta^{(\alpha)}} [c_1 P_\theta (X)^{1 - \alpha} f_1 (X) + \cdots +
c_k P_\theta (X)^{1 - \alpha} f_k (X)] = 0.
\end{eqnarray}
Therefore,
\begin{equation*}
c_1 P_\theta (X)^{1 - \alpha} f_1 (X) + \cdots + c_k P_\theta (X)^{1 - \alpha} f_k (X) = t,
\end{equation*}
where $t = \mathbb{E}_{\theta ^{(\alpha)}} [c_1 P_\theta (X)^{1 - \alpha} f_1 (X) + \cdots +
c_k P_\theta (X)^{1 - \alpha} f_k (X)]$. 
\begin{equation}
\label{1eqn:linear_combination_of_f_i}
c_1 f_1 (\cdot) + \cdots + c_k f_k (\cdot) = t  P_\theta (\cdot)^{\alpha-1}.
\end{equation}
As $P_\theta \in \mathbb{M}^{(\alpha)}$, we have
\begin{equation*}
P_\theta (\cdot)^{\alpha-1} = Z(\theta)^{1-\alpha} [Q(\cdot)^{\alpha -1} + (1 - \alpha)\theta ^T f (\cdot)],
\end{equation*}
for some $\theta = (\theta_1,\dots, \theta_k)\in \Theta$. Substituting this in (\ref{1eqn:linear_combination_of_f_i}), we get
\begin{eqnarray}
\label{1eqn:linear_combination_f_i_and_Q}
[c_1 - (1-\alpha)t \theta _1 Z(\theta)^{1 - \alpha} ] f_1 (\cdot) + \cdots + 
[c_k - (1-\alpha)t \theta _k Z(\theta)^{1 - \alpha} ] f_k (\cdot) \nonumber\\ 
- t Z(\theta)^{1 - \alpha}  Q(\cdot)^{\alpha -1} = 0.
\end{eqnarray}
Observe that, with out loss of generality, we can assume that the vectors $f_1(\cdot),\dots, f_k(\cdot)$ in the $\alpha$-power-law family are linearly independent. Further, in view of footnote \ref{1ftnote:linear_independence}, we claim that the vector $Q(\cdot)^{\alpha-1}$ is linearly independent of the vectors $f_1(\cdot),\dots, f_k(\cdot)$. If not, let $Q(\cdot)^{\alpha-1} = c_1 f_1(\cdot) + \cdots + c_k f_k (\cdot)$ for some non-zero vector $c = (c_1, \dots, c_k)^T$. 
Then, for any $P\in \mathbb{L}$, $\sum\limits_{x\in\mathcal{X}}P(x)Q(x)^{\alpha-1} = \sum\limits_{x\in\mathcal{X}} P(x)[c_1 f_1(x) + \cdots + c_k f_k (x)] = 0$. This is not possible since $Q$ has full support. Then, from (\ref{1eqn:linear_combination_f_i_and_Q}), we have $t = 0$ and hence $c_1 = \cdots = c_k = 0$, which is a contradiction. Hence $A$ is positive definite. This completes the proof.

\section*{Acknowledgements}
Atin Gayen is supported by an INSPIRE fellowship of the Department of Science and Technology,
Government of India.

\end{document}